%% file: main.tex
\documentclass[preprint,11pt,3p]{elsarticle}

\usepackage{url, array, calc, color}
\usepackage{algorithm, algorithmic}
\usepackage{graphics, epstopdf, graphicx, epsfig, psfrag, subfigure}
\usepackage{amssymb, amsmath}
\usepackage{multicol,balance, verbatim}
\usepackage{mdwmath, mdwtab, stfloats}
\biboptions{sort&compress}
\usepackage{picins}

\usepackage{amsthm}

\newtheorem{theorem}{Theorem}[section]
\newtheorem{lemma}[theorem]{Lemma}

\newtheorem{example}{Example}

\newcommand{\ie}{{\em i.e., }}
\newcommand{\eg}{{\em e.g., }}

\begin{document}

\begin{frontmatter}

\title{Active Topology Inference using Network Coding}

\author[uci]{Pegah Sattari}
\ead{psattari@uci.edu}
\author[epfl]{Christina Fragouli}
\ead{christina.fragouli@epfl.ch}
\author[uci]{Athina Markopoulou}
\ead{athina@uci.edu}

\address[uci]{Calit2 Building, Suite 4100, UC Irvine, Irvine, CA 92697-2800, United States}
\address[epfl]{EPFL IC ARNI, Building BC, Station 14, CH - 1015 Lausanne, Switzerland}

\begin{abstract}
Our goal is to infer the topology of a network when (i) we can send probes between sources and receivers at the edge of the network and (ii) intermediate nodes can perform simple network coding operations, \ie additions. Our key intuition is that network coding introduces topology-dependent correlation in the observations at the receivers, which can be exploited to infer the topology. For undirected tree topologies, we design hierarchical clustering algorithms, building on our prior work in \cite{allerton}. For directed acyclic graphs (DAGs), first we decompose the topology into a number of two-source, two-receiver (2-by-2) subnetwork components and then we merge these components to reconstruct the topology. Our approach for DAGs builds on prior work on tomography \cite{journal}, and improves upon it by employing network coding to accurately distinguish among all different 2-by-2 components. We evaluate our algorithms through simulation of a number of realistic topologies and compare them to active tomographic techniques without network coding. We also make connections between our approach and alternatives, including passive inference, {\tt traceroute}, and packet marking.
\end{abstract}


\end{frontmatter}

\section{Introduction}

Knowledge of network topology is important for network management,
diagnosis, operation, security and performance optimization. Depending on the context, one may be interested in the topology at different layers, such as the Internet's router-level topology, an overlay network topology, the topology of an ad-hoc wireless network, etc. 
 
There is a large body of prior work on measurements and inference of network topology. One family of techniques assumes the cooperation of nodes in the middle of the network, and uses {\tt traceroute} measurements to collect the ids of nodes along paths and use them to reconstruct the topology. Another family of techniques, referred to as {\em network tomography}, assumes no cooperation from internal nodes and relies on end-to-end probes to infer internal network characteristics, including topology. More specifically, multicast or unicast probes are sent/received between sets of sources/receivers at the edge of the network, and the topology is inferred based on the number and order of received probes. 
 
In this paper, we revisit the problem of topology inference using end-to-end probes, in networks where intermediate nodes are equipped with simple network coding capabilities. We show how to exploit these capabilities in order to perform active topology inference in a more accurate and efficient way than existing tomographic techniques. 

Our key intuition is that network coding introduces topology-dependent correlation in the content of packets observed at the receivers, which can then be exploited to reverse-engineer the topology. For example, a coding point (that combines multiple incoming packets into one or more outgoing packets) introduces correlation between packets coming from different sources, in a similar way that multicast introduces correlation in the packets sent by the same source and observed by several receivers. In fact, the correlation introduced by multicast has been the starting point and the main idea underlying tomographic topology inference. Subsequent schemes made this idea more practical, by emulating multicast with back-to-back unicast probes \cite{castroUnicast,probing}. In contrast, relating probes from different sources to reveal intermediate nodes, also referred to as multiple-source tomography, has been a challenge \cite{bu,probing,journal}. Using simple network coding operations at coding points solves this problem and allows accurate and fast topology inference.

Our approach is general and can be applied to infer the topology in a range of scenarios, including but not limited to wireless multi-hop networks. Wireless multi-hop networks are able to support simple network coding operations  (additions  are sufficient for our schemes), as demonstrated in \cite{cope}, and can therefore benefit from our techniques. Furthermore, there is a good match between some properties and constraints of such networks and our schemes. First, there is natural variability in the delay of wireless links, which (if appropriately used - as explained in later sections) can expedite inference. Second, our schemes keep internal nodes simple (moving processing for inference to dedicated nodes at the edge) and anonymous (revealing the logical topology but not the identities of nodes). Finally, improving the speed of inference may prove important to keep up with changes, {\em e.g.,} due to mobility. 

Our contributions are as follows. First, we consider undirected trees, where leaves can act as sources or receivers of probes, and we design hierarchical clustering algorithms that infer the topology, building on our prior work in \cite{allerton}. Then, we consider directed acyclic graphs (DAGs) with a fixed set of M sources and N receivers and a pre-determined routing scheme. We first decompose the topology into a number of two-source, two-receiver subnetwork components and then we merge these components to reconstruct the topology. Our approach for DAGs builds on prior work on tomography \cite{journal}, and improves upon it by employing simple network coding operations at intermediate nodes to deterministically distinguish  among all possible 2-by-2 subnetwork components, which was impossible without network coding \cite{journal,probing}. We evaluate our algorithms through simulation over a number of topologies  and we show that they can infer the topology accurately and faster than tomographic approaches without network coding. We present our schemes as active probing:  special probes are sent by the sources, specifically for the purpose of inference,  and are treated in special ways by intermediate nodes and eventually received by the receivers and processed at a fusion center. We believe that our  active probing approach with network coding provides one more building block, in the already large space of topology inference techniques, with core strength and ability to identify joining points. We also compare and make connections between our active probing approach and alternatives, such as passive inference, {\tt traceroute}, and packet marking.

The rest of the paper is organized as follows. Section~\ref{sec-related} discusses related work. Section~\ref{sec-statement} presents our assumptions, notation, and problem statement. Section~\ref{sec:MainResults} summarizes the main results of the paper. Section~\ref{sec-trees} presents algorithms for inferring tree topologies. Binary trees are discussed in Section~\ref{sec-binaryTree}, in the absence (Section~\ref{sec-binaryLossless}) or presence (Section
\ref{sec-binaryLossy}) of packet loss. General trees are discussed in Sections~\ref{sec-m-aryTree} and \ref{sec-generalTree}. Section~\ref{sec-extension} presents algorithms for inferring directed acyclic graphs (DAGs). Section~\ref{sec-2by2} presents algorithms for inferring 2-by-2 subnetwork components, in the absence (\ref{sec-2by2Lossless}) or presence (\ref{sec-2by2Lossy}) of packet loss. Section \ref{sec-merging} explains how to merge these components to reconstruct the topology. Section~\ref{sec-simul} provides simulation results for some realistic topologies. Section~\ref{sec-connections} discusses two possible deployment scenarios (one as an active probing scheme and another one using packet marking), and makes connections  between our approach and alternative topology inference approaches. Section~\ref{sec-conclusion} concludes the paper. Appendices A and B analyze the probability of error of our inference algorithms in trees and DAGs, respectively.

\section{Related Work}
\label{sec-related}

One body of related work is network tomography in general, and topology inference in particular. A  good survey of network tomography can be found in \cite{survey}. An early work on topology inference using end-to-end measurements is \cite{sylvia}, where the correlation between end-to-end multicast packet loss patterns was used to infer the topology of binary trees.  The correctness of this idea was rigorously established in \cite{duffield2002}, and was extended to general trees and to measurements other than loss, such as delay variance \cite{topologyDelay}, or more generally any metric that grows monotonically with the number of traversed links. The idea has also been extended to unicast probes \cite{castroUnicast, probing}. In summary, tomographic schemes for topology inference use end-to-end active probing and feed the number, order, or a monotonic property of received probes as input to statistical signal-processing techniques. Inference of link characteristics \cite{minc} can also be combined with topology inference \cite{probing}. In a different context, similar problems have been studied in the context of {\em phylogenetic} trees \cite{phylogenetic}. The work in \cite{anima} uses such algorithms \cite{phylogenetic}, for topology inference in sparse random graphs. 

In addition, inference of congested links has been studied from the angles of compressive sensing \cite{infocom2011, coatesCS1, coatesCS2} and group testing \cite{cheraghchi, gtBook, gtThiran}. The work in \cite{cheraghchi} formulates the problem as a {\em graph-constrained} group testing, where the items correspond to edges, some of them being defective, and the goal is to identify the defective edges given that the test matrix conforms to constraints imposed by the graph, \eg the path connectivities. The work in \cite{infocom2011} recovers sparse vectors, representing certain parameters of the links over the graph, through $l_1$ minimization. It improves the number of required measurements over \cite{cheraghchi}, as  compressive sensing allows real numbers for the link characteristics and measurements instead of true/false binary values in group testing problems.

Most tomographic approaches rely on probes sent from a single source in a tree topology \cite{sylvia, metric-induced, castroUnicast, adaptive, mixture, classicMethod, duffield2002, striped, caceresTopology, topologyDelay, eriksson2010}. Rabbat {\em et al.} \cite{journal, probing, merging} introduced the multiple-source multiple-destination (M-by-N) tomography problem, by sending probes between $M$ sources and $N$ receivers. In \cite{journal,probing}, it was shown that an M-by-N network can be decomposed into a collection of 2-by-2 components. Then, coordinated transmission of multi-packet probes from two sources and packet arrival order measurements at the two receivers were used to infer some information about the 2-by-2 topology. Assuming knowledge of M 1-by-N topologies and all 2-by-2 components, it was also shown how to merge a second source's 1-by-N tree topology with the first one. The resulting M-by-N topology is not exact, but bounds are provided on the locations of joining points with respect to the branching points. This approach also requires a large number of probes, as do all approaches that need to collect enough probes for statistical significance \cite{topologyDelay, adaptive, striped, castroUnicast, radar}. Our work on DAGs builds on and extends the multiple-source multiple-destination work in \cite{journal,probing}, but uses network coding to achieve exact and fast topology inference. 

A second body of related work is from the network coding literature. It is well known that linear network coding makes a network behave as a linear system, whose transfer function depends on the topology. Based on the source packets and the observations at the receivers, one can then try to passively infer the topology. The following papers consider that random linear network coding is employed for the purpose of information transfer, and they perform {\em passive} inference on the side. In \cite{ho}, passive techniques are used to distinguish among failure patterns. In \cite{p2p, subspace, jaggi, jaggi-journal}, subspace properties at various nodes are used for topology inference and error localization. In \cite{jaggi}, each node passively infers its upstream topology at no cost to throughput, but at high complexity.

In contrast, we propose {\em active} probing and a simple coding scheme at intermediate nodes, to achieve low-complexity topology inference at the end nodes. In Section~\ref{sec-connections}, we provide a detailed comparison and make connections between active and passive topology inference. In \cite{globecom, linkloss-journal}, we revisited link-loss (but not topology) tomography using active probing and network coding. In the first part of this paper, we extend our preliminary work in \cite{allerton}, where we showed that active probes from two sources and \texttt{XOR} at intermediate nodes are sufficient to infer binary tree topologies. This approach generalizes to trees, but not to general graphs. In \cite{netCod}, we used a different approach for general graphs, which builds on \cite{journal, probing}: we identify 2-by-2 components and merge them together in an M-by-N topology. This journal paper combines and extends our preliminary work in \cite{allerton,netCod}.

A practical approach for inferring the network topology is based on  {\tt traceroute}  \cite{govindan,cheswick,rocketfuel,yao,clauset,dallasta,skitter,doubletree,paris,DIMES}.
Multiple {\tt traceroute}'s are sent among monitoring hosts, they record node ids along paths,
 and this information is put together to reconstruct the graph. The {\tt traceroute}-based approach is
 discussed in detail in Section~\ref{sec-traceroute}. 
 
Wireless sensor networks and information fusion are considered in  \cite{loss-baochun,lossNC-baochun}. Information is collected at sensor nodes and is forwarded towards a fusion center, following a known reverse tree topology. Information is aggregated \cite{loss-baochun} or network coded \cite{lossNC-baochun} at intermediate nodes, and the loss rates  of links are inferred based on the observations at the fusion center. In contrast, we are interested in inferring the topology of DAGs.

\section{Problem Statement}
\label{sec-statement}

\subsection{Model}

{\bf Assumptions about the Network.} We are interested in inferring {\em static} topologies\footnote{Our algorithms assume that the topology remains static during the inference. However, the topology may change over longer time scales.}. We are also interested in inferring {\em logical} topologies, which are defined by the branching and joining points where the measured end-to-end paths
meet\footnote{Intermediate nodes in a logical topology have degree at least 3, and in-degree and out-degree at least 1. Since it is necessary for identifiability, focusing on logical topologies
is a standard assumption in topology inference problems.}.

In the first part of the paper, we consider undirected trees with $|V|=n$ vertices, $|E|=n-1$ edges that can be used in both directions, and exactly one path between any two vertices. We denote by $\mathcal{L}=\{1,2,...,L\}$ the leaf-vertices of the tree, which correspond to end-hosts that can act as sources or receivers of probe packets.

In the second part, we consider directed acyclic graphs (DAGs) with M sources and N receiver nodes, which we refer to as {\em M-by-N topology}, following the terminology of \cite{journal, probing}. Without loss of generality (W.l.o.g.), we present most of our discussion in terms of $M=2$, \ie inferring a 2-by-N topology; an M-by-N topology can be constructed by merging smaller structures. Similarly to \cite{journal, probing}, we also assume that a predetermined routing policy maps each source-destination pair to a unique route from the source to the destination.\footnote{Our assumption for single path routing is based on the following reasons: (1) single path routing is the reality in most networks today: routers pick the unique next hop towards the destination; (2) this was also the assumption made by Rabbat {\em et al.} \cite{journal, probing}, which is the starting point on which we build in this paper, by adding simple network coding operations; (3) topology inference when multi-path routing is used is an open problem;  the state of the art is [1], which proposes a heuristic approach; and (4) network coding is only used on special probes for the purpose of inference here, and not for throughput, which could be improved by using multi-path routing.} This implies the following three properties, first stated in \cite{journal}.
\begin{itemize}\addtolength{\itemsep}{-.35\baselineskip}
\item[A1] There is a unique
path from each source to each receiver.
\item[A2] Two paths from the
same source to different receivers take the same route until they
branch, so that all 1-by-2 components have the ``inverted Y''
structure; the node where the paths to the two receivers split is
called a {\em branching point}, $B$.
\item[A3] Two paths from different sources
to the same receiver use exactly the same set of links after they
join, so that all 2-by-1 components have the ``Y'' structure; the
node where the paths from the two sources merge is called a {\em
joining point}, $J$.
\end{itemize}
These properties are consistent with destination-based routing: the next hop taken by a packet is determined by a routing table lookup on the destination address. Each subnetwork from one source to $N$ receivers is a 1-by-N tree; the general graph is called a ``multiple-tree'' network \cite{journal}.

{\bf Loss and Delay.}
We consider scenarios with and without packet loss. Each link has a delay with a fixed part, \eg the propagation and transmission delay, and a variable part, \eg the queueing delay. Path delay is the sum of delays across the links in the path. We have no control over the delays of the links but we have control over the timing of operations at sources and intermediate nodes. We can make sources and intermediate nodes operate in time slots of duration $T$ and $W$, respectively, which can be chosen to be quite longer than link delays as explained later.\footnote{This can be
achieved, for example, by assuming a coarse synchronization across source and network nodes (on the order of 5-10ms using NTP), and by making nodes wait for a window before sending or coding/forwarding probe packets respectively.}

{\bf Goal.} Our goal, in this paper, is to design active probing schemes, \ie the operation of sources, intermediate nodes
and receivers, that will allow us to infer the logical topology from the observations at the receivers. We restrict the space of possible operations to the simple options described in the rest of this section. In later sections, we design schemes based on these simple operations and we show that they are sufficient for topology inference. We will revisit the problem statement and make it more precise in the sections for trees and DAGs.

{\bf Operation of Sources.} An experiment consists of a pair of sources $S_1$ and $S_2$ sending, at the same time, a multicast probe packet each ($x_1=[1,0]$ and $x_2=[0,1]$, or more generally symbols from a finite field) to all $N$ receivers. These are special probes sent {\em solely} for the purpose of inference, not for regular data transfer, and treated in a special way, specified next, by intermediate nodes. We perform up to $countMax$ experiments. Consecutive experiments are spaced apart by a large time interval $T$, to ensure that only probes in the same experiment are combined together.

{\bf Operation of Intermediate Nodes.} Intermediate nodes are assumed to support unicast, multicast
and the simplest possible network coding operation, \ie addition over a finite field $\mathbf{F}_q$. They operate in time slots of pre-determined duration or window $W$: a node waits for
$W$ to receive probe packets from its incoming links; if it receives more than one packet, it codes them together and forwards (unicast or multicast) the resulting packet downstream.
The choice of $W$ affects where the packets from the two sources meet. Essentially, an intermediate node can act either as a {\em joining point} (J), in which case it adds all incoming packets and forwards the output to all outgoing links;\footnote{In our schemes, all joining points perform network coding. Therefore, in the rest of the paper, we use the terms {\em joining point}, which comes from the tomography literature \cite{merging,probing,journal}, and {\em coding point}, which comes from the network coding literature, interchangeably.} or as a {\em branching point} (B), in which case it sends (multicasts) the single incoming packet downstream. This operation will be specified more precisely in the sections for trees and DAGs.
%\textcolor{blue}{For example, as we are going to see in Section~\ref{sec-extension}, by choosing $W$ to be greater than the maximum path delay, we can make sure that packets meet, in a hop-by-hop manner, and are coded together.}

{\bf Operation of Receivers.} Each receiver $i$ receives probes $R_i$, which are the source packets  $x_1$,  $x_2$, or a linear combination of $x_1$ and $x_2$, as the result of network coding operations at intermediate nodes. Inference of topology is based only on the observations $R_i$'s. We assume that these observations are sent to a fusion center for central processing and inference; consistently to all tomography literature, the communication of the receivers and the fusion center is out of the scope of this paper.

{\bf Intuition.} Multicast as well as network coding (which is limited to
simple addition in this paper, thus can be thought of as reverse multicast) introduce topology-dependent correlation in the content of packets, which can be used at the receivers to infer the underlying topology. In particular, multicast helps reveal the branching points while network coding helps reveal the joining points.

\subsection{Scope and Discussion}

Possible deployment scenarios are described in Section \ref{sec-deployment}. The first scenario (sending special probes for the sole purpose of topology inference) is used to describe the schemes throughout the paper. Furthermore, and beyond the specific  details of the deployment, we believe that our work provides a fundamental building block for exploiting correlation in the content of network coded packets for inference of joining points. Similarly,  multicast tomography showed how to exploit correlation in multicast packets for inference of branching points; it was then followed by a series of papers that used unicast traffic to ``mimic'' multicast probes and the whole functionality while being more practical.

We would like to emphasize that, in this paper, we apply  network coding on special probes solely for the purpose of topology inference, and {\em not} for improving data transfer (decoding the source messages at the receivers). In data transfer, throughput and delay are indeed important metrics. In our problem, the important metrics are: identifiability (for which, we show that network coding is necessary); and efficiency, \ie the number of probes used and the amount of network resources consumed for a certain level of estimation accuracy  (which we show that it is improved by network coding). Therefore, the delay of the algorithms is not of primary concern in this paper: inferring the topology in the order of seconds as opposed to milliseconds is acceptable in our setup. Multi-path routing, which could increase throughput with network coding, is not considered either.

\section{Main Results}
\label{sec:MainResults}

The main results obtained in this paper are the following:
\begin{itemize}\addtolength{\itemsep}{-.35\baselineskip}
\item For tree networks: 
\begin{itemize}
\item When there is no packet loss in the links, we design the deterministic Alg.~\ref{alg-binaryLossless}, which infers the topology in $O(n)$ iterations, where $n$ is the number of edges in the tree.
\item When there is packet loss in the links, we design Alg.~\ref{alg-binaryLossy}, which infers the topology in $O(nM)$ iterations, where $M\leq \frac{1}{P_{min}}$, and $P_{min}$ is the minimum probability of success across all paths between a source and a destination.
\end{itemize}
\item For DAGs, we decompose the topology into 2-by-2 subnetwork components, and then we merge these components to reconstruct the topology. We design inference algorithms to infer the 2-by-2's, and we design merging algorithms to merge them back to the original topology:
\begin{itemize}
\item When there is no packet loss in the links, we design Alg.~\ref{alg-generalLossless}, which identifies any 2-by-2 topology with probability of error analyzed in Eq.(\ref{eq-error-Alg5}), in $countMax$ experiments. 
\item When there is packet loss in the links, we design Alg.~\ref{alg-generalLossy}, which identifies any 2-by-2 topology with probability of error analyzed in Lemma~\ref{thm-Thm3}, in $countMax$ experiments. 
\item Assuming knowledge of all the 2-by-2's and one source's 1-by-N tree topology, we design Alg.~\ref{alg-merging1} that merges a second source's topology with the first one by identifying {\em all} the joining points in $O(N\log N)$ steps.
\item Assuming knowledge of all the 2-by-2's, but not the 1-by-N tree topology, we can identify all the joining points if and only if there are no branching points in a row. We merge the two topologies in $\binom{N}{2}$ steps, as we describe in Section~\ref{sec-merging2}.
\item We also provide a lower bound on the number of 2-by-2's required by any merging algorithm to uniquely localize all the joining points in a 2-by-N topology, given one source's 1-by-N topology. In Lemma~\ref{theorem-minimum2by2s}, we show that it is $\frac{N}{2}$. 
\end{itemize}
\item We also make connections between our approach and alternative topology inference approaches in Section~\ref{sec-connections}.
\end{itemize}

Note that Alg.~\ref{alg-generalLossless}, Alg.~\ref{alg-generalLossy}, and Alg.~\ref{alg-merging1} build on and extend the corresponding algorithms by Rabbat {\em et al.} \cite{journal, probing, merging}, in the presence of network coding.

\section{Inferring Trees}
\label{sec-trees}

{\bf Overview.} We design algorithms for inferring undirected tree topologies, based only on probes sent between leaf nodes. We follow a hierarchical, top-down approach, by iteratively dividing the tree topology into smaller clusters and revealing how the groups are connected to each other.

{\bf Operation of Sources and Receivers.} In each iteration (timeslot $T>>W$) a set of leaves (different across timeslots) are chosen to act as sources and the remaining leaves act as receivers. Each source sends a distinct packet.  The receiver stores the first packet it receives, and discards any subsequent packets (in the same iteration).

{\bf Operation of Intermediate Nodes.} Every intermediate node operates in intervals of duration $W$. If, within $W$, the node receives a single probe from one of its neighbors, it multicasts the probe to all other neighbors. If, within $W$, it receives more than one packet from different neighbors, it adds them and forwards the result to all remaining neighbors. In binary trees, this linear combination is simply \texttt{XOR}. In general trees, we need operations over higher fields.\footnote{Note that other mappings at the joining point, from $(x_1,x_2)$ to $f(x_1,x_2)$, would achieve the same goal. Linear network coding, $f(x_1,x_2)= x_1+x_2$ , is only one such mapping. Concatenation $(x_1,x_2)$, for example, is another mapping. This approach is, for example, used as data aggregation in \cite{loss-baochun}, where a node waits to receive data from all its children in the reverse multicast tree (or until a specified period of time has elapsed). The node then aggregates all the data and forwards it to the sink via the reverse tree. However, the advantages of using network coding, in particular, compared to these other mapping functions, include simple linear operations and fixed packet size. Indeed, when we have more than two source packets meeting at a joining point, network coding provides an output packet of fixed size at the output, while concatenation provides an output packet with output linear in the number of incoming packets. The same advantage applies when we have the same probe packet meeting multiple times with the other probe at a $J$ (as it may be the case in DAGs), \eg network coding results in $2x_2$. In summary, although it is possible to use other approaches for the same purpose, network coding is the most efficient way to do the task.}

{\bf Summary of Results.} In the rest of this section,  we first consider binary trees, with or without packet loss. Then we extend our algorithms to $m$-ary trees. For trees without loss, we design deterministic algorithms that infer the topology in $O(n)$ iterations. For trees with loss, just one successfully received probe per network path is sufficient, without the need to collect packet loss statistics, a property that enables rapid discovery of the underlying topology.

\subsection{Binary Trees} \label{sec-binaryTree}

\subsubsection{Lossless Binary Tree} \label{sec-binaryLossless}

Let us first consider the simplest case: an undirected binary tree without packet loss. The following example illustrates the main idea.
\begin{example}\label{ex:binarytree}
Consider the tree shown in Fig.~\ref{fig-binaryExample}, with 7 leaves (1,2, ...7) and 5 intermediate nodes (A,B,C,D,E). Assume that nodes $1,7$ act as
sources $S_1,S_2$  and send probes $x_1=[1,0],x_2=[0,1]$, respectively. All other leaves act as receivers. Intermediate node $A$ receives $x_1$
and forwards it to leaf $2$ and to node $C$.
Similarly, node $D$ receives $x_2$ and
forwards it to node $E$ (which in turn forwards it to leaves $5,6$) and to node $C$. Probe packets $x_1$ and $x_2$ arrive at node $C$, which adds them, creates the packet $x_3=x_1\oplus
x_2=[1,1]$, and forwards $x_3$ to node $B$,  which in turn forwards it
to leaves $3,4$.\footnote{We have chosen the directionality of the edges depending on which source reaches the intermediate node first. In this example, we assume that all links have the same delay. For different delays,  $x_1,x_2$ could meet at different nodes, but the algorithm will still work, as discussed after Lemma \ref{thm-Thm1}.}
 
 \begin{figure*}[t!]
\subfigure[Undirected binary tree we want to infer.]{\includegraphics[width=4.4cm, height=4.4cm]{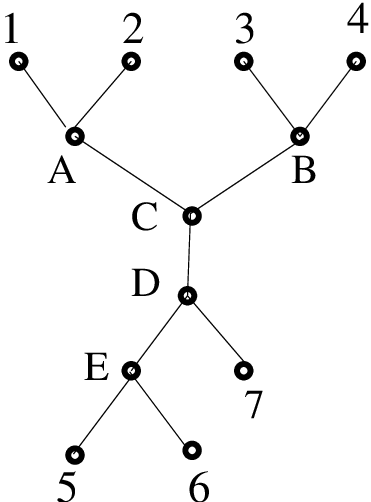}\label{fig-binaryExample}}
\hspace{1cm}
\subfigure[Structure revealed after 1 iteration. Leaves 1 and 7 act as sources. Probes meet at $C$.]{\includegraphics[width=4.4cm, height=4.4cm]{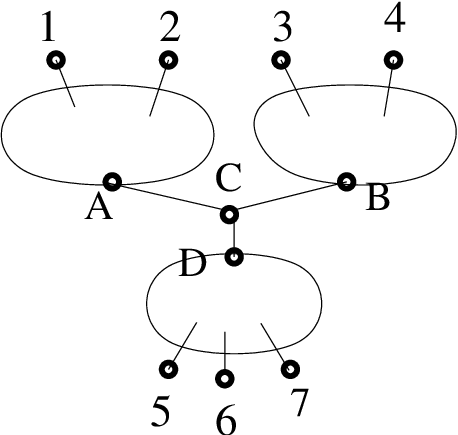}\label{fig-binaryExample2}}
\hspace{1cm}
\centering \subfigure[Structure revealed after two iterations. Leaves 5 and 6 act as sources. Probes meet at  $E$.]{\includegraphics[width=4.4cm, height=4.4cm]{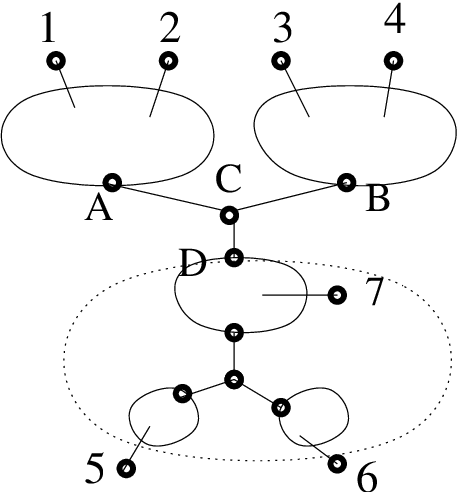}\label{fig-binaryExample3}}
\caption{\label{fig-revealing} Example \ref{ex:binarytree}: inferring the topology of an undirected binary tree  with 7 leaves and 5 intermediate nodes.} 
%\vspace{-1.0em}
\end{figure*}

At the end, leaf $2$ receives $x_1$, leaves $5,6$
receive $x_2$ and leaves $3,4$ receive
$x_3=x_1\oplus x_2$. Thus, the leaves of the tree can be partitioned into three sets:
$\mathcal{L}_1$ containing $S_1$ and the leaves that received $x_1$, \ie  $\mathcal{L}_1=\{1,2\}$; $\mathcal{L}_2=\{5,6,7\}$ containing $S_2$ and the leaves
that received $x_2$; and $\mathcal{L}_3=\{3,4\}$
containing the leaves that received $x_1\oplus x_2$.
From this information observed at the edge of the network, we can
deduce that the binary tree has the structure depicted in
Fig.~\ref{fig-binaryExample2}: three components, each seeing a different probe ($x_1, x_2, x_1\oplus x_2$) flowing through it, and connected through three links to the middle node $C$. This concludes the first experiment/iteration.

To infer the structure that connects leaves $\{5,6,7\}$ to node $C$,
we need a second experiment. We randomly choose
two of these leaves, {\em e.g.,} $5,6$, to act as sources $S_1,S_2$. Any probe packet leaving node $D$ will be multicast to all remaining leaves of the tree, {\em
i.e.,} nodes $\{1,2,3,4\}$ observe the same packet. One can think of node $D$ as a single
``aggregate-receiver'' that observes
the common packet received at nodes $\{1,2,3,4\}$. Following the
same procedure as before, assuming that $x_1,x_2$ meet
at node $E$, nodes $7$ and $\{1,2,3,4\}$ receive $x_3=x_1\oplus x_2$. Using this additional information and the fact that the topology
is a binary tree, we refine the inferred structure from Fig.~\ref{fig-binaryExample2} to Fig.~\ref{fig-binaryExample3}. \end{example}

Algorithm~\ref{alg-binaryLossless} generalizes the previous example and can infer any binary tree topology. It starts by considering all the leaves $\mathcal{L}$. It calls {\tt SendTwoProbes} and partitions $\mathcal{L}$ into smaller sets $\mathcal{L}_1$, $\mathcal{L}_2$, $\mathcal{L}_3$. It proceeds by recursively calling {\tt SendTwoProbes} within each set, until all edges are revealed.

\begin{algorithm}[t!]
\begin{footnotesize} 
\caption{\label{alg-binaryLossless} Topology Inference for Lossless Binary Tree}
\begin{algorithmic}[1]
\STATE $E=\emptyset$; /*Initially, we can only observe the leaves ($\mathcal{L}$); our goal is to reveal all the edges in the tree, \ie set $E$.*/
\STATE {\bf {\tt InferBinaryTree($\mathcal{L}$):}}
\STATE ($\mathcal{L}_1, \mathcal{L}_2, \mathcal{L}_3, A_1, A_2, A_3$)={\tt SendTwoProbes}($\mathcal{L}$);    %/*$\mathcal{L}$ is the set of all leaves in the tree*/
\FOR{$i\in\{1,2,3\}$}
\IF{$|\mathcal{L}_i|==0$}
\STATE Continue;
\ELSIF{$|\mathcal{L}_i|==1 ~||~ |\mathcal{L}_i|==2$} 
\FOR{$v\in \mathcal{L}_i$} 
\STATE $E=E\cup \{(v,A_i)\}$; /*Connect the leaf nodes $v$ in $\mathcal{L}_i$ through node $A_i$ to the rest of the network.*/
\ENDFOR
\ELSE 
\STATE /*$|\mathcal{L}_i|>2$, \ie $\mathcal{L}_i$ contains three or more leaves.*/
\STATE {\bf {\tt InferBinaryTree($\mathcal{L}_i \cup A_i$)}}; /*Node $A_i$ that connects $\mathcal{L}_i$ to the network acts as an aggregate receiver\footnotemark.*/
\ENDIF
\ENDFOR
\RETURN
\STATE Replace vertices of degree two with a single node.
\vspace{1.0em}
\STATE {\bf {\tt SendTwoProbes}($\mathcal{L}$):}
\STATE Randomly choose two leaves in $\mathcal{L}$ to act as sources $S_1$, $S_2$ and
send probe packets $x_1$, $x_2$ respectively. All other leaves $\mathcal{L}-\{S_1,S_2\}$ act as receivers. Intermediate nodes act as branching or joining points. %in trees.
\STATE When all receivers receive a probe, partition $\mathcal{L}$ into
$\mathcal{L}_1,\mathcal{L}_2,\mathcal{L}_3$ as follows. 
\STATE Set $\mathcal{L}_1$ contains $S_1$ and all receivers that observe $x_1$. Set $\mathcal{L}_2$ contains $S_2$ and all receivers that observe $x_2$. Set $\mathcal{L}_3$ contains all receivers that observe $x_3=x_1\oplus x_2$.
\IF{$|\mathcal{L}_3|\neq 0$}
\STATE Create new nodes $A_1,A_2,A_3,A_*$.
\STATE $E=E\cup \{(A_1,A_*),(A_2,A_*),(A_3,A_*)\}$; 
\STATE /*This case is depicted in Fig.~\ref{fig-dividing}(a): $\mathcal{L}$ is divided into three components $\mathcal{L}_1$, $\mathcal{L}_2$, $\mathcal{L}_3$, connected through three edges to nodes $A_1,A_2,A_3$.*/
\ELSE 
\STATE /*$\mathcal{L}_3=\emptyset$*/
\STATE Create new nodes $A_1,A_2$. Also, $A_3=null$.   
\STATE $E=E\cup\{(A_1,A_2)\}$; 
\STATE /*This case is depicted in Fig.~\ref{fig-dividing}(b): $\mathcal{L}$ is divided into 2 components $\mathcal{L}_1,\mathcal{L}_2$, connected through a single edge.*/
\ENDIF 
\RETURN the components $\mathcal{L}_1, \mathcal{L}_2, \mathcal{L}_3$ and the nodes $A_1, A_2, A_3$. %that connect $\mathcal{L}_i$ to the network.
\end{algorithmic}
\end{footnotesize}
\end{algorithm}

\begin{figure}[t]
\centering 
\subfigure[Dividing $\mathcal{L}$ into 3 components.]  
{{\includegraphics[height=4.2cm,width=4.7cm]{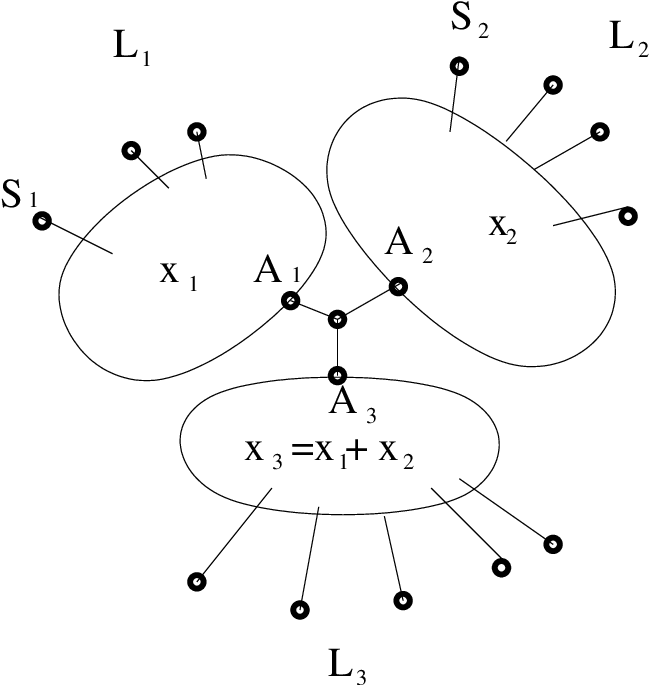}} \label{fig-case3}} \hspace{0.5cm} 
\subfigure[Dividing $\mathcal{L}$ into 2 components]
{{\includegraphics[height=2.4cm,width=4.6cm]{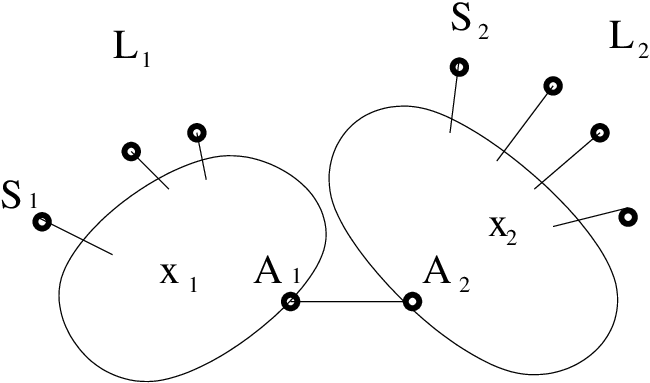}}
\label{fig-case21}}
\caption{Edges and vertices of the graph, as revealed by a single iteration (call of {\tt SendTwoProbes}) in Algorithm \ref{alg-binaryLossless}. The leaves $\mathcal{L}$ are partitioned into two or three groups, based on their observations, $x_1, x_2, x_1\oplus x_2$.} 
\label{fig-dividing}
\end{figure}

\begin{lemma}
\label{thm-Thm1} Algorithm~\ref{alg-binaryLossless} terminates in at most $n$ iterations and exactly infers the topology of an undirected binary tree.
\end{lemma}

\begin{proof} 
Consider a particular iteration (call of {\tt SendTwoProbes}): sources $S_1$ and $S_2$ send exactly one probe packet each to all other leaves. Now consider the intermediate nodes on the path $\mathcal{P}$ between
the two sources. Depending on the link delays, there are two possibilities.

The first possibility is that $x_1$ and $x_2$  meet (arrive within the
same $W$) at one of the intermediate nodes on $\mathcal{P}$, {\em e.g.,} node
$A$. Node $A$ forwards their \texttt{XOR} to its third link, and the iteration reveals the neighboring edges and nodes to
$A$ as depicted in Fig.~\ref{fig-dividing}(a).
Another possibility is that $x_1$ and
$x_2$ cross each other while traversing the same link of
$\mathcal{P}$ in opposite directions, \ie they do not meet
at a node. Even if a leaf node receives more than one
probe, we design their operation so that they only keep the first one. In this
case, we infer the configuration in Fig.~\ref{fig-dividing}(b) that reveals one edge.

In summary, the algorithm iteratively divides the binary tree into smaller components until one component has two or less leaves, in which case we know its structure. In each iteration, we reveal three edges or one edge. At the end, we have revealed all $n-1$ edges. Therefore, the algorithm requires between $\frac{n-1}{3}$ and $n-1$ iterations. 
\end{proof}

\footnotetext{Although we cannot directly observe $A_i$, whatever is received by $A_i$ will be received by the leaves that are in $\mathcal{L}$ but not in $\mathcal{L}_i$; thus acting as an ``aggregate'' receiver on their behalf.}

{\em Notes.} In each iteration, every link is traversed exactly once  by a probe.
Link delays affect where the probes meet and thus what components are revealed in each iteration. However, they do
 not affect the correctness of the algorithm.

\subsubsection{Lossy Binary Tree} \label{sec-binaryLossy}

Packet loss may cause confusion when dividing the receivers into components. One solution is to send multiple probes from the same two sources in each iteration as we discuss next. However, given packet loss and delay variability, this may result in probes meeting at different nodes in the same iteration\footnote{This was not a problem in the lossless case. In a given iteration, since only one probe packet is generated by each source, the packets at most meet at one intermediate node.}. This problem exists because we deal with undirected graphs, where a link may be traversed in opposite directions by probes sent in the same iteration. It can be avoided by fixing the directionality of edges in each iteration. This can be achieved in a distributed way by the first packet arriving at each intermediate node. %as described in the following.
We modify the intermediate node operations as follows.

{\bf Intermediate Node Operation:} Each intermediate node keeps a
table of its neighbors. In each iteration, it marks these
neighbors as source or sink
neighbors. Once this marking is done, it does not change
for the duration of the iteration. The first time during an iteration that an
intermediate node receives a probe,
%\footnote{We may use a special type of probe packet for the beginning of an iteration.},
it waits for a window $W$ to receive probes
from other neighbors. After this time $W$ passes,
the node marks all neighbors from which it received packets as
sources and all other neighbors as sinks. For the remaining duration
of the iteration, the node accepts packets only if they originate
from its source neighbors. If the node receives a packet from one of its source neighbors, it forwards it to all its sink neighbors. If it receives more than one packet from different source neighbors, it linearly combines them, and forwards
the result to its sink neighbors. The node rejects probes coming
from sinks, and does not forward packets towards sources.

\begin{algorithm}[t!]
\begin{footnotesize} 
\caption{\label{alg-binaryLossy} Topology Inference for Lossy Binary Tree. We only describe the {\bf {\tt SendTwoProbes}} function below since the first part (which contained the {\bf {\tt InferBinaryTree}} procedure) is similar to Algorithm~\ref{alg-binaryLossless}. {\bf {\tt SendTwoProbes}} is different from Algorithm~\ref{alg-binaryLossless} in that: (i) each source sends $M$, instead of one, probe packets; and (ii) the rules to divide $\mathcal{L}$ into three components change as follows.}
\begin{algorithmic}[1]
\STATE {\bf {\tt SendTwoProbes}($\mathcal{L}$):}
\STATE Randomly choose two leaves in $\mathcal{L}$ to act as $S_1$, $S_2$. The sources transmit, for $M$ times, probe packets $x_1$, $x_2$, respectively. All other leaves $\mathcal{L}-\{S_1,S_2\}$ act as receivers. Intermediate nodes act as branching or joining points. 
\STATE /*Partition $\mathcal{L}$ into $\mathcal{L}_1,\mathcal{L}_2,\mathcal{L}_3$ as follows.*/
\FOR{each receiver $j$}
\STATE Let $O_j$ be the set of  all observations of receiver $j$. /*We consider the union of observations for each receiver.*/
\STATE /*For aggregate receiver $A_i$, we apply the same rule using the union of the
aggregate receiver observations.*/
\IF{$O_j$ contains only $x_1$} 
\STATE Assign receiver $j$ to the set $\mathcal{L}_1$.
\ELSIF{$O_j$ contains only $x_2$} 
\STATE Assign receiver $j$ to the set $\mathcal{L}_2$.
\ELSIF{$O_j$ contains both $x_1$ and $x_2$, or it contains an $x_1\oplus x_2$ packet} 
\STATE Assign receiver $j$ to the set $\mathcal{L}_3$.
\ELSE
\STATE /*$O_j=\emptyset$*/
\STATE Randomly assign receiver $j$ to one of the components.
\ENDIF
\ENDFOR
\IF{$|\mathcal{L}_3|\neq 0$}
\STATE Create new nodes $A_1,A_2,A_3,A_*$.
\STATE $E=E\cup \{(A_1,A_*),(A_2,A_*),(A_3,A_*)\}$; 
\STATE /*This case is depicted in Fig.~\ref{fig-dividing}(a): $\mathcal{L}$ is divided into three components $\mathcal{L}_1$, $\mathcal{L}_2$, $\mathcal{L}_3$, connected through three edges to nodes $A_1,A_2,A_3$.*/
\ELSE 
\STATE /*$\mathcal{L}_3=\emptyset$*/
\STATE Create new nodes $A_1,A_2$. Also, $A_3=null$.   
\STATE $E=E\cup\{(A_1,A_2)\}$; 
\STATE /*This case is depicted in Fig.~\ref{fig-dividing}(b): $\mathcal{L}$ is divided into 2 components $\mathcal{L}_1,\mathcal{L}_2$, connected through a single edge.*/
\ENDIF 
\RETURN the components $\mathcal{L}_1, \mathcal{L}_2, \mathcal{L}_3$ and the nodes $A_1, A_2, A_3$.
\end{algorithmic}
\end{footnotesize}
\end{algorithm}

Alg.~\ref{alg-binaryLossy} presents the modifications required for Alg.~\ref{alg-binaryLossless} to be able to infer binary trees with lossy links. The main difference is that in each iteration, each source sends $M$ instead of one probes.

{\em Performance:} Alg.~\ref{alg-binaryLossy} has an associated probability of error, since a leaf might not receive the ``correct'' probe packet\footnote{In a given iteration, an error may occur either because a leaf does not receive any packet (which can be made arbitrarily small by increasing the number of probes $M$) or, because it belongs to $\mathcal{L}_3$ but happens to receive only $x_1$ or only $x_2$ packets. This probability decreases very fast as $M$ increases, as observed in the simulations of Section \ref{sec-simul}.}. For our algorithm to operate correctly, it suffices that each receiver receives  {\em at least one} packet from each of the sources it is connected to. Nodes in $\mathcal{L}_1$ or $\mathcal{L}_2$ are connected to one source ($S_1$ or $S_2$), while nodes in $\mathcal{L}_3$ to two sources. In general, the number of probes $M$ required per iteration in order to have one ``success'' is a random variable that depends on the topology. For general trees, the distribution of this random variable is difficult to characterize precisely, but upper bounds can be provided. In particular, we need every path (from each source to each receiver) to work at least once. Different paths have different probability of success depending on their length ($l_P$) and the probability of loss on every link across the path ($p_i$ for link $i$): $P=(1-p_1)\cdots(1-p_{l_P})$. Let $P_{min}$ be the minimum probability of success across all paths. Then $M$ is a geometric random variable with success probability $P_{min}$. Therefore, $E[M]=1/{P_{min}}$, $var[M]=(1-P_{min})/P^2_{min}$, and $Pr(|M-E[M]|\geq m)\leq var[M]/m^2$. An example computation for the exact probability of error in a specific topology (tree of Fig.~\ref{fig-binaryExample}) is provided in Appendix A. Note that in general, $M$ is much smaller compared to the methods that collect a statistically significant number of packets and perform estimation.

%Note that the number of packets $M$ required to infer the topology within a desired error probability, is much smaller compared to the methods that collect a statistically significant number of packets and perform estimation. 
%
%The upper bound on $M$ is determined by the longest path, of length $l_{max}$: $M\leq \frac{1}{\overline p^{l_{max}}}$. 

\subsection{M-ary Trees}

\subsubsection{Full M-ary Trees} \label{sec-m-aryTree}

We first consider full $m$-ary trees, where all intermediate nodes have degree $m+1$, $m \ge 3$,
without packet loss.  Alg.~\ref{alg-binaryLossless} can still
accurately infer the topology in less than $n$ iterations. However, we can modify the algorithm to infer the topology even faster. The idea is to keep the hierarchical clustering approach but increase the number of components revealed in each iteration, either (i) by
changing the intermediate nodes so that
they forward different linear combinations of incoming probes
to different outgoing links; or (ii) by increasing the number of
sources  in each iteration.

{\bf Modification I:} (two sources per iteration, coding points send different combinations to different links).
When an intermediate node receives two incoming packets from two different neighbors, it deterministically generates different linear combinations,
{\em e.g.,} $x_1+x_2,x_1+2x_2,\cdots$ and forwards each resulting
packet to a different neighbor.
Therefore, when $x_1$ and $x_2$ meet at any intermediate node, the leaves of the network
will be divided into $m+1$ components, depending on which probe packet they receive. If the probe packets do not meet at a node but cross each other, the leaves of the network will be divided into two
components. Once a component has $m$ or less leaves, since we have a full $m$-ary tree, we know its structure. Therefore, in each iteration, we reveal $m+1$ edges or one edge, and the total number of iterations is reduced to at least $\frac{n-1}{m+1}$ and at most $n-1$ iterations.
Note that the operations are performed over $F_{m^2}$ in this case.

{\bf Modification II:} (more than two sources per iteration, coding points send the same combination to all outgoing links). Alternatively, we can use up to $m$ sources (as per Lemma \ref{thm-Thm2}) per iteration. The sources send $x_1=[1,0,0,\cdots,0],x_2=[0,1,0,\cdots,0],\cdots,x_m=[0,0,0,\cdots,1]$, respectively. When an intermediate node receives $k$ packets from
different neighbors within $W$, $2\leq k\leq m$, it simply
adds them up (over $F_{2^m}$) and forwards the result to all
remaining neighbors. Depending on whether the
node receives $k$ packets or only a single packet, the
leaves of the network will be divided into $m+1$ or $m$ more
components; \ie in each iteration, we reveal $m+1$ or $m$ edges. Therefore, the algorithm requires at least $\frac{n-1}{m+1}$ and at most $\frac{n-1}{m}$ iterations.

\begin{lemma}
\label{thm-Thm2} The maximum number of sources that can be used to uniquely infer the topology of a full $m$-ary tree is $m$.
\end{lemma}

\begin{proof}
We show that if we use $m+1$ sources to infer the topology of a full $m$-ary tree, it cannot be uniquely identified. For example, consider a binary tree with three sources sending
$x_1=[1,0,0]$, $x_2=[0,1,0]$, and $x_3=[0,0,1]$ respectively, to all other leaves in the tree. Assume that the three probe packets meet at one intermediate
node; thus, we divide the leaves into four components, which observe $x_1$, $x_2$,
$x_3$, and $x_1+x_2+x_3=[1,1,1]$ respectively. Since the degree of
intermediate nodes is three, we conclude that two of the three sources must have joined at
one intermediate node first, and then their result must have joined with the
third source in another intermediate node, so that they result in
$x_1+x_2+x_3$ in the last component. The first two sources can be
either $x_1, x_2$ or $x_2, x_3$. Therefore, we cannot uniquely
infer the underlying binary tree topology by observing these four
components. The same discussion applies to larger full $m$-ary trees ($m>2$).
\end{proof}

{\em Note.} In the presence of loss, the same
argument as in Section~\ref{sec-binaryLossy} applies, \ie we can
assign directions to edges in each iteration, so that
our algorithms are applicable to the lossy case as well.

\subsubsection{General M-ary Trees} \label{sec-generalTree}

In general $m$-ary trees, the degree of intermediate nodes varies from three up to a maximum of
$m+1$. We can still apply Alg.~\ref{alg-binaryLossless} and infer the tree topology in $O(n)$ iterations. We can also apply Modification I described in
Section~\ref{sec-m-aryTree}; the operations are still performed over $F_{m^2}$ since probe packets may meet at an intermediate node of degree $m+1$. However, we cannot apply Modification II here: since probe packets may meet at an intermediate node of degree three, we cannot use more than two sources according to Lemma \ref{thm-Thm2}, although there exist larger degree nodes in the tree.

\section{Inferring Directed Acyclic Graphs (DAGs)}
\label{sec-extension}

\subsection{From a Single-Tree to Multiple-Tree Topologies}

So far, we considered undirected
trees. Let us now consider directed trees, which are a special case of DAGs.
\begin{example}
Assume that we assign
directions to the links of the binary tree in
Fig.~\ref{fig-binaryExample}, all from the top to the bottom. Clearly, we can no longer send probe packets in arbitrary directions in each iteration.  However, we can still infer some
information about the topology. Assume that we send probes
 from the source nodes $1$ and $2$, and we observe $x_1\oplus
x_2$ at the receiver nodes $5$, $6$, and $7$. Therefore, we
identify three components $\mathcal{L}_1=\{1\}$,
$\mathcal{L}_2=\{2\}$, and $\mathcal{L}_3=\{5,6,7\}$, together with the
intermediate nodes $A$ and $D$, and three edges $1A$, $2A$, and $AD$,
which connect the three components together. However, we cannot
obtain more information about the internal structure of the
component $\mathcal{L}_3$ or any other part of the tree network.
\end{example}

Next, consider a 2-by-2 network as defined in Section \ref{sec-statement}, \ie a directed acyclic graph with two sources, two receivers and predetermined routing. Note that directed trees are only one type among all four possible types of the basic 2-by-2 components of any multiple-tree network, as defined in Section~\ref{sec-statement}.  There exist four 2-by-2 topologies, as shown in Fig.~\ref{fig-2by2}, which were first defined in \cite{journal, probing}. Following the same terminology as in \cite{journal, probing},  we refer to Fig.~\ref{fig-2by2}(a), (b), (c) and (d) as type 1, 2, 3 and 4,
respectively. Type 1 is called {\em shared} \cite{journal,
probing} since the joining points for both receivers coincide
($J_1=J_2$) and the branching points for both sources coincide
($B_1=B_2$). The other three types (types 2, 3 and 4) are called
{\em non-shared} since they have two distinct joining points and two
distinct branching points.

\begin{figure*}[t!]
\centering \subfigure[type 1: shared]{\includegraphics[scale=0.2]{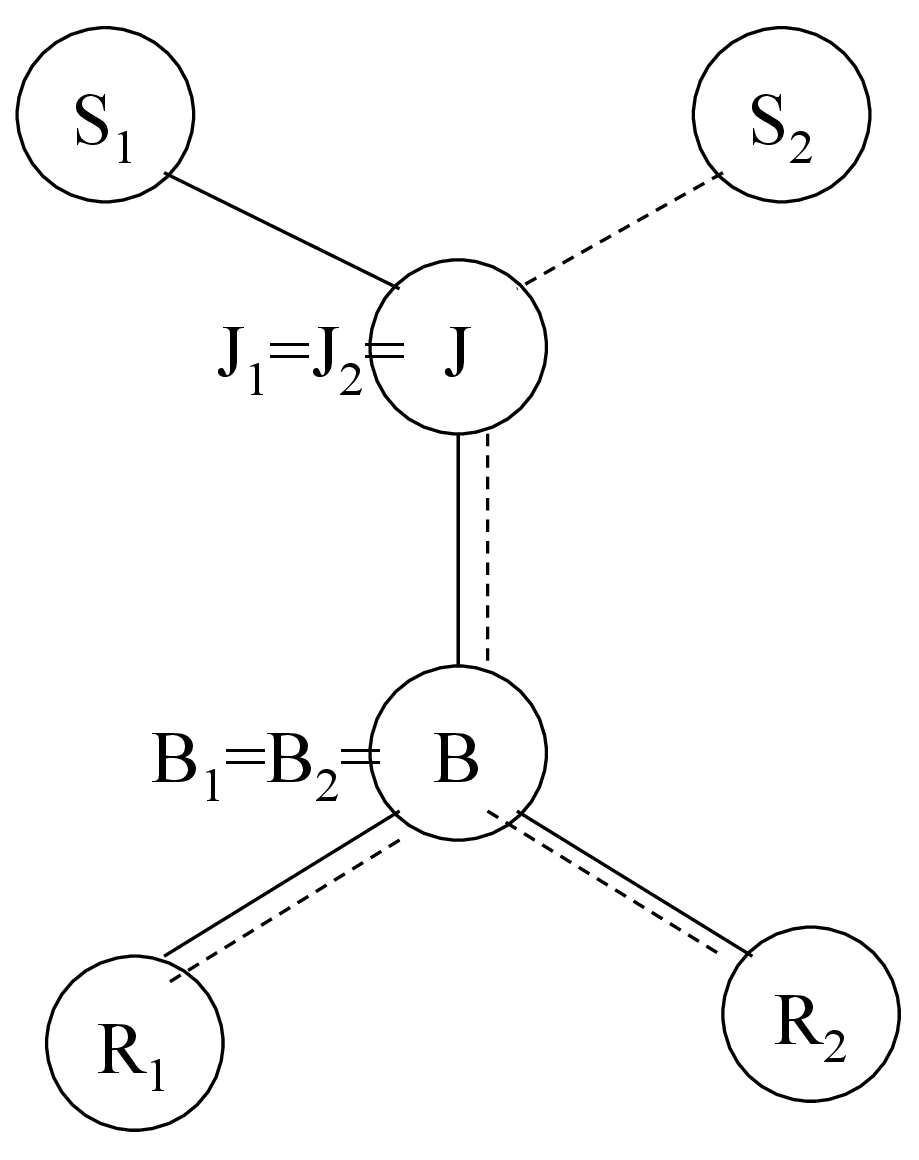}} \hspace{0.5cm} \subfigure[type 2: non-shared]{\includegraphics[scale=0.2]{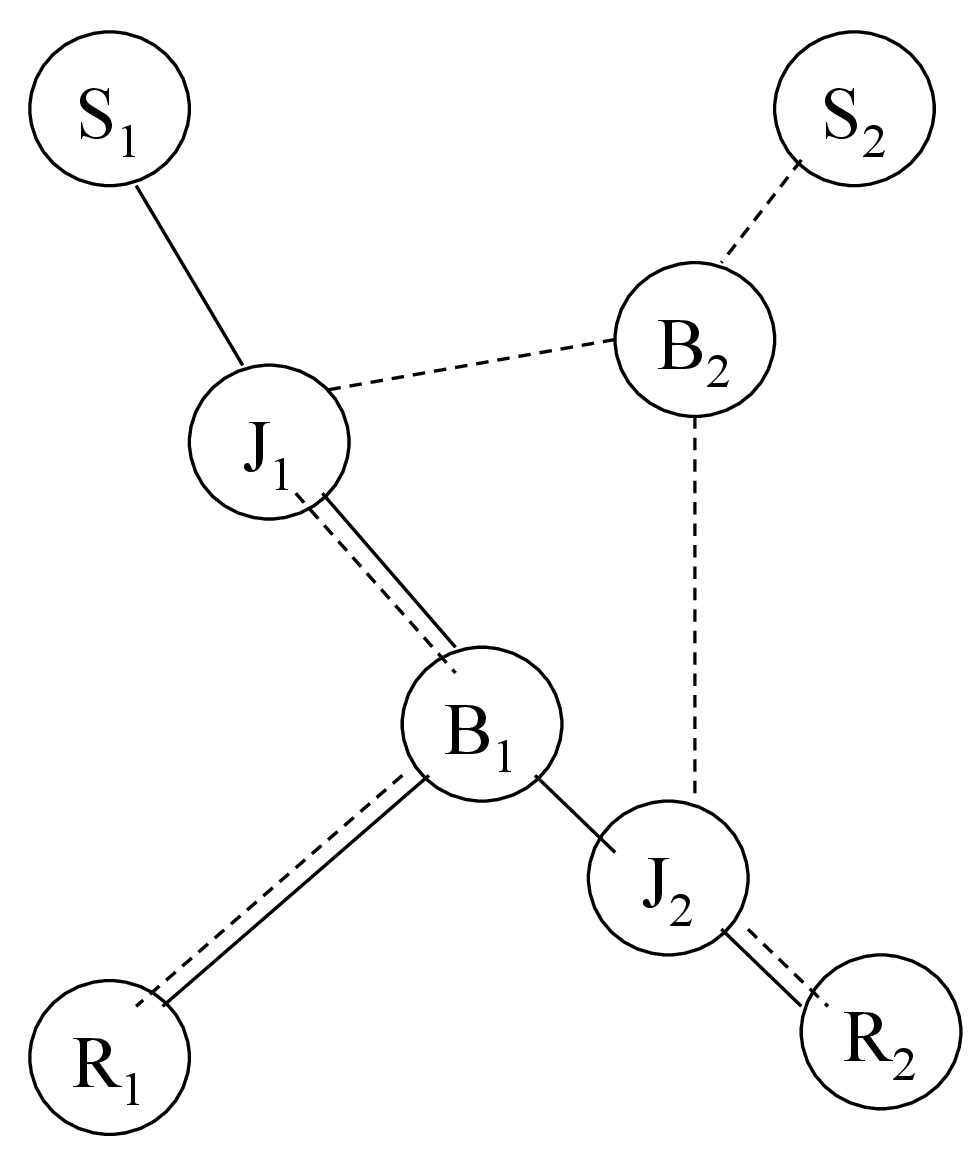}} \hspace{0.5cm} \subfigure[type 3: non-shared]{\includegraphics[scale=0.2]{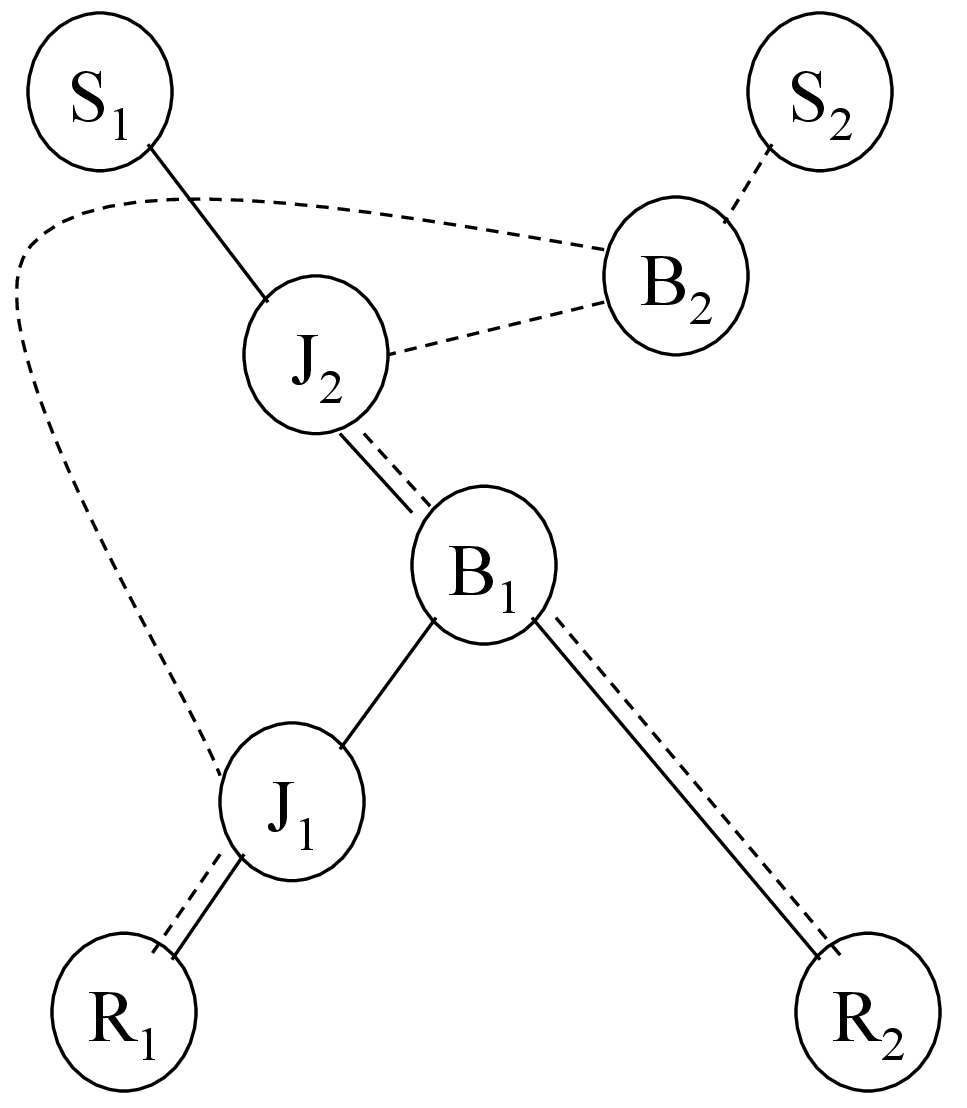}} \hspace{0.5cm} \subfigure[type 4: non-shared]{\includegraphics[scale=0.2]{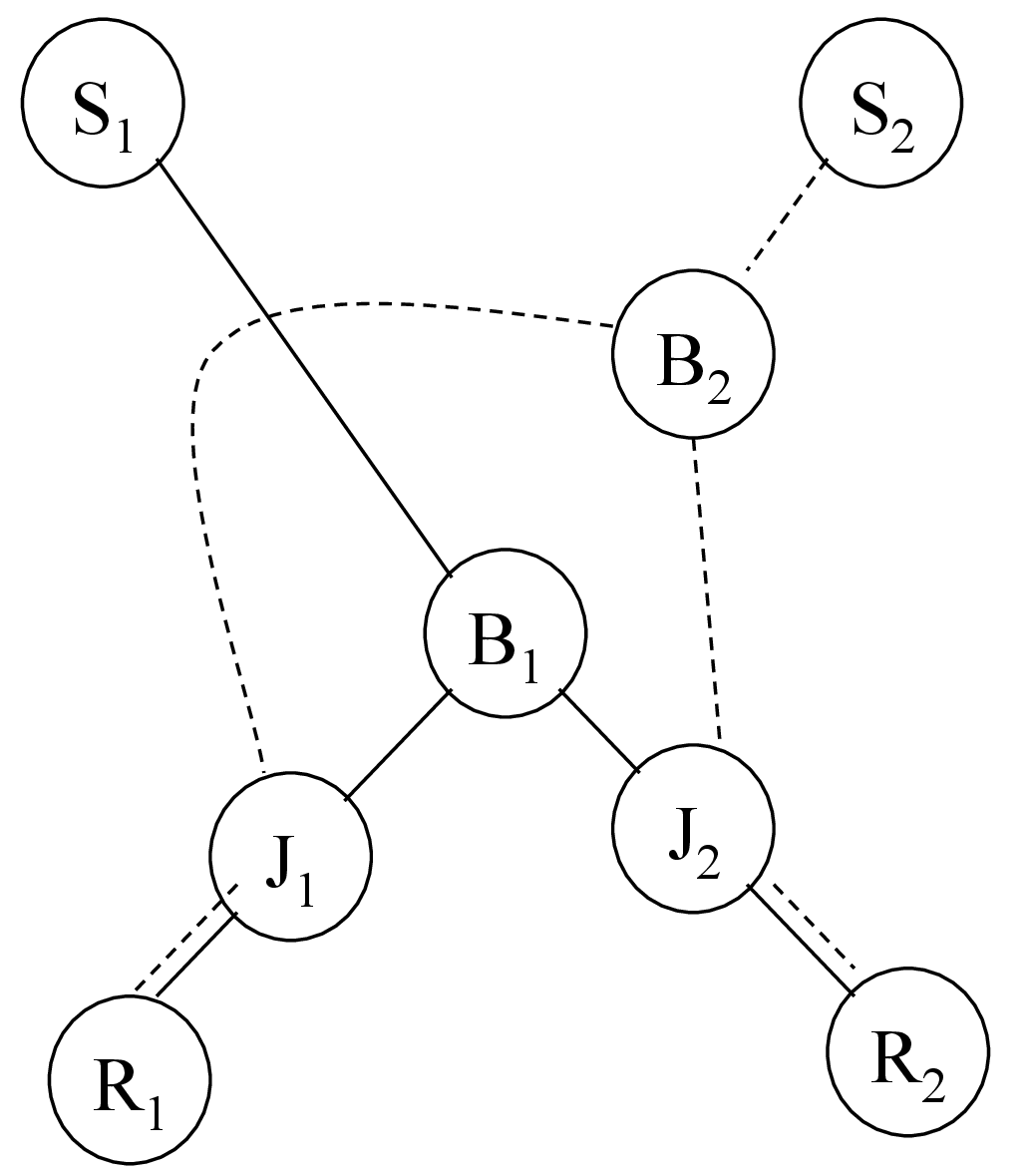}} \caption{The four possible types of a 2-by-2 subnetwork component, as defined in \cite{journal}. There are two sources ($S_1,S_2$) multicasting packets $x_1,x_2$ to two receivers ($R_1, R_2$). (The 1-by-2 topology of $S_1$ is a tree composed of $S_1, B_1, R_1, R_2$. Similarly, The 1-by-2 tree rooted at $S_2$ is $S_2, B_2, R_1, R_2$. $J_1$ and $J_2$ are the joining points, where the paths from $S_2$ to $R_1$ and $R_2$, join/merge with $S_1$'s topology.)}
\label{fig-2by2}
\end{figure*}

In a directed tree, all 2-by-2 components are of type 1. However, in a
general M-by-N topology, several different 2-by-2 types may co-exist.
The algorithms described so far can identify type 1 2-by-2
topologies, and thus, trees (either completely or partially, as
described above). However, they cannot distinguish between type 1
and type 4 2-by-2's, as described in the following example.
\begin{example}
Consider Fig. \ref{fig-2by2} (a)
and (d). Assume that in both cases, we send $x_1,x_2$ from $S_1,S_2$
to $R_1,R_2$ and that $x_1,x_2$ meet (arrive within the same $W$) at any joining point. Therefore, in both type 1 and type 4,
both receivers observe $x_1+x_2$, and we cannot distinguish between the two types.
\end{example}
In general, unlike single-tree networks, the observations do not uniquely characterize
the underlying topology in multiple-tree networks. The reason is that once two sources in a tree network transmit their probe packets, they at most meet at one coding point for all the receivers, as we saw in Section~\ref{sec-trees}. On the other hand, in a multiple-tree network, probe packets may meet at different coding points for different receivers, as depicted in Fig.~\ref{fig-extension}. Therefore, we need a different approach.

\begin{figure}[t]
\begin{center} {\includegraphics[scale=0.4]{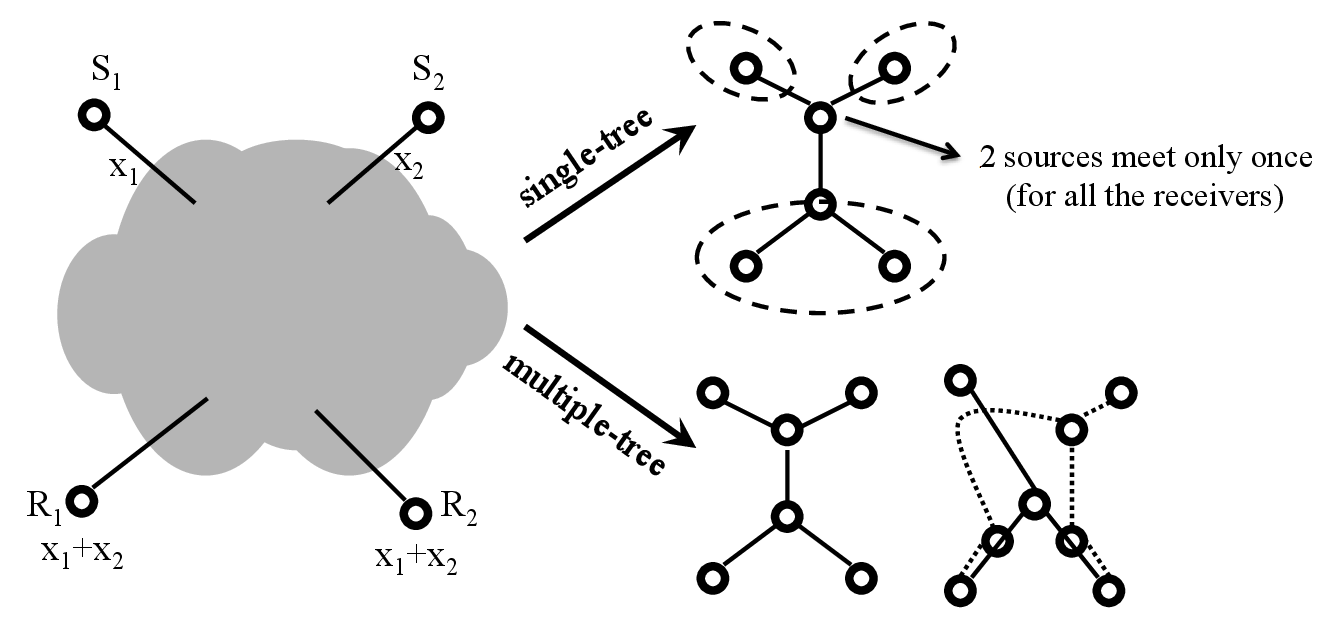}}
\vspace{-0.5em}
\caption{Single-tree vs. multiple-tree
topologies. Consider a single iteration. In a multiple-tree topology, unlike the single-tree topology, the observations at the receivers no longer uniquely identify the topology.} 
\label{fig-extension}
\vspace{-1.0em}
\end{center}
\end{figure}

{\bf Problem Statement.} Our goal in this section is to infer a multiple-tree topology, or an ``M-by-N'' topology according to the terminology of Section \ref{sec-statement}. Similarly to \cite{journal}, we take two steps. In the first step (Section \ref{sec-2by2}), we use several experiments and we exactly identify the type of every 2-by-2 component. In the second step (Section \ref{sec-merging}), we merge these 2-by-2 subnetwork components to reconstruct the M-by-N network.

{\bf Operation of Sources.} Pairs of sources are selected and send up to $countMax$ coordinated multicast packets to all receivers. As in the general setup, probes are spaced apart by intervals of length $T$. In addition, we introduce a difference in the sending time of the two sources, which we call the offset $u$. W.l.o.g., let  $S_1$ send first and $S_2$ second.

The timing parameters $T, u, W$ are coarsely tuned so as to create observations that can distinguish among different 2-by-2 types. In particular, (i) $T>> W$ ensures that only probes within the same experiment are coded together. To be more precise, we choose $T\geq g\cdot W$, where $g$ is the maximum number of joining points on any $(S_i,R_j)$ path in the topology. In the worst case, there can be $N$ joining points in a row and thus, $g\leq N$. However, in practice, $g$ is usually a lot smaller. (ii) $W>>$ path delay (between the sources and the joining points) ensures that source packets meet at the joining points despite link delays. (iii) $u$ is selected randomly in each iteration, so that it forces probes to meet at different points, or not meet at all, in different iterations. Finally, coarse selection of $T,W$ with rough estimates of upper bounds on link and path delays is sufficient.

{\bf Operation of Receivers.}
For a given 2-by-2 subnetwork, let the observations at the two receivers  be $R_1=c_{11}x_1+c_{12}x_2$, $R_2=c_{21}x_1+c_{22}x_2$. Based on these observations, we design {\em Inference} algorithms that identify the 2-by-2 type (in Section \ref{sec-2by2}) and {\em Merging} algorithms that build the M-by-N from the 2-by-2's (in Section \ref{sec-merging}).

{\bf Operation of Intermediate Nodes.} In DAGs, the operation of an intermediate node, depending on whether it acts as a joining point or a branching point, is summarized in Alg.~\ref{alg-joiningPoint} and Alg.~\ref{alg-branchingPoint}, respectively. A joining
point (J) adds and forwards packets, while a branching point (B)
forwards the single received packet to all ``interested'' links downstream. A link is ``interested'' in the routing sense if it is the next hop for at least one source packet in the network coded packet.

\subsection{Identifying  2-by-2 Components}
\label{sec-2by2}

In this section, we propose an approach to exactly identify a 2-by-2 component, using the same intuition as in trees, \ie coding operations result in observations that can uniquely characterize the underlying 2-by-2.  Our approach builds on \cite{journal} and improves over it by uniquely distinguishing among all four 2-by-2 types, while \cite{journal} could only distinguish between shared and non-shared types.

\begin{algorithm}[t!]
{\footnotesize \caption{\label{alg-joiningPoint} \footnotesize {\bf Operation at
Joining Point $J$, in DAGs.} When two sources multicast to N receivers, $J$
has two incoming links and one outgoing link.
Additions are performed over $\mathbf{F}_q$.}
\begin{algorithmic}[1]
\FOR{every time window $W$} \IF{($J$ receives two packets within $W$
from its incoming links)} \STATE as soon as the last packet arrives, it adds
them up, and forwards the resulting packet downstream. \ELSIF{($J$
receives only one packet within $W$)} \STATE it forwards the packet
downstream. \ELSIF{($J$ does not receive any packet within $W$)}
 \STATE /*nothing to do.*/
\ENDIF \ENDFOR
\end{algorithmic}
}
\end{algorithm}

\begin{algorithm}[t!]
{\footnotesize \caption{\label{alg-branchingPoint} \footnotesize {\bf Operation at
Branching Point $B$, in DAGs.} While two sources multicast to N receivers,
$B$ has one incoming packet and multiple outgoing links.}
\begin{algorithmic}[1]
\FOR{each incoming packet} \IF{the incoming packet is $x_1$ (or
$x_2$)} \STATE forward it only on the outgoing links that are next
hops for $S_1$ ($S_2$). \ELSE \STATE /* The incoming packet is of the
form $ax_1+bx_2$. */ \STATE forward the packet to all outgoing links.
\ENDIF \ENDFOR
\end{algorithmic}
}
\end{algorithm}

\subsubsection{Lossless 2-by-2} \label{sec-2by2Lossless}

First, we provide an algorithm to identify the type of a 2-by-2 component without packet loss.
In the first experiment, sources $S_1,S_2$ multicast probe packets
$x_1, x_2$ to $R_1,R_2$. We begin with the assumption that $S_1,
S_2$ act simultaneously, or in practice within the synchronization
offset. A choice of large $W$ guarantees that $x_1,x_2$ meet at both joining points $J_1,J_2$, which add the incoming probes over $\mathbf{F}_3$.
Depending on the underlying 2-by-2 type, $R_1, R_2$ observe one of the following pairs:
\begin{itemize}\addtolength{\itemsep}{-.35\baselineskip}  
\item type 1:   $R_1$: $x_1+x_2$   ,  $R_2$: $x_1+x_2$
\item type 2:   $R_1$: $x_1+x_2$   ,  $R_2$: $x_1+2x_2$
\item type 3:   $R_1$: $x_1+2x_2$  ,  $R_2$: $x_1+x_2$
\item type 4:   $R_1$: $x_1+x_2$   ,  $R_2$: $x_1+x_2$
\end{itemize}

Types 2 and 3 result in unique observations that make them
distinguishable from any other type; \ie one such observation suffices to identify
type 2 or type 3. However, types 1 and 4 result in the same pair of observations; therefore,
we need to design different experiments to get observations that can uniquely characterize type 1 or type 4.

In the next experiment, we exploit the observation, first made in \cite{journal},
that type 1 is the only 2-by-2 where the two joining points coincide ($J_1=J_2=J$).
Therefore, the observations at the two receivers are always the
same: either $x_1+x_2$ when the two packets meet at $J$; or a single
packet ($x_1$ or $x_2$) when the two packets do not meet at $J$. In
contrast, type 4 has two different joining points $J_1 \neq J_2$. If
we force packets to meet only at one of the joining points but not
at the other one, the receivers will have different observations. These
are observations \#3 and \#4 in Table~\ref{table-losslessObs} and
they can uniquely characterize type 4.

These observations can be achieved by appropriately selecting the offset $u$ in
the sources' sending times. $u$
needs to be large enough so that after addition to the link delays,
it can affect $W$: if $D_1,D_2$ represent the 
delays on the paths from $S_2$ to $J_1,J_2$, respectively, $u$
must be in $[W-D_1, W-D_2]$\footnote{\label{footnote-offset} In 2-by-2
components, this interval is close to $W$, since $D_1$ and $D_2$ are small compared to $W$. In more general 2-by-N networks that we consider for our simulations, there exist multiple
links between the sources and the joining points, link delays are on the
order of tens of ms, and $W$ is on the order of hundreds of ms.
Therefore, we can safely choose $u \in [f\cdot W,W]$ in the general
case, where $0<f<1$ is a tunable parameter. We choose
$f=0.7$ in our simulations, to force different observations
at the two receivers.}.\footnote{In fact, we can obtain similar observations without using the offset $u$, but instead, by changing $W$ in a range of values, from the maximum path delay (as it is now), down to 0. One can check that Alg. 5 (without $u$) can still be applied in this case. Therefore, using the offset is not really crucial in our scheme.}

\begin{table}[t!]
{\footnotesize \caption{\label{table-losslessObs}{\bf
Lossless Case.} Possible observations for types 1 and 4 2-by-2
topologies. (Observation \#1 occurs when the sources are
synchronized. Observations \#2-4 occur when $S_2$ sends after $S_1$, by an offset
$u\in [f\cdot W,W]$.)}
\vspace{-1.0em}
\begin{center}
\begin{tabular}[t]{|c||c|c||c|c|}
\hline
Observation        &   \multicolumn{2}{|c||}{Type 1}      &      \multicolumn{2}{|c|}{Type 4}              \\
Number             &    $\mathbf{R_1}$  & $\mathbf{R_2}$   &      $\mathbf{R_1}$    &    $\mathbf{R_2}$        \\
\hline \hline
1                         &         $x_1+x_2$   &       $x_1+x_2$  &            $x_1+x_2$    &    $x_1+x_2$   \\
\hline \hline
2                         &         $x_1$       &      $x_1$       &             $x_1$       &    $x_1$       \\
\hline
3                         &                     &                  &           $x_1+x_2$     &    $x_1$        \\
\hline
4                         &                     &                  &            $x_1$        &    $x_1+x_2$            \\
\hline
\end{tabular}
\end{center}
}
\vspace{-1.5em}
\end{table}

\begin{algorithm}[t!]
{\footnotesize \caption{\label{alg-generalLossless} \footnotesize {\bf Lossless
Case - Inferring a 2-by-2 component.} Sources $S_1$ and $S_2$ multicast
$x_1$ and $x_2$. Receivers observe $R_1=c_{11}x_1+c_{12}x_2$ and
$R_2=c_{21}x_1+c_{22}x_2$.}
\begin{algorithmic}[1]
\STATE n=1; /*first experiment*/ \IF{$c_{22}>c_{12}$} \STATE Output
type 2; \ELSIF{$c_{22}<c_{12}$} \STATE Output type 3; \ELSE \STATE
/*It is $R_1=R_2$*/ \WHILE{$n<countMax ~\& ~ R_1==R_2$} \STATE Draw
offset $u$ uniformly at random out of $[f\cdot W,W]$. \STATE Send
probes; $S_2$ transmits $u$ time later than $S_1$. \IF{$R_1\neq
R_2$} \STATE Output type 4; \STATE Exit; \ENDIF \STATE n++; \ENDWHILE
\STATE Output type 1; /* It is $n==countMax$*/ \ENDIF
\end{algorithmic}
}
\end{algorithm}

Alg.~\ref{alg-generalLossless} summarizes the experiments we perform in order to infer the type of
a 2-by-2 network. Types 2 and 3 are identified in the first observation.
Type 4 is identified the first time that the two receivers see different observations.
If after $countMax$ trials, we still have not seen any different observations at the two receivers, then we declare the 2-by-2 to be of type 1.

{\em Choosing $countMax$.} $countMax$ should be large enough to ensure small probability of error. The probability of error of Alg.~\ref{alg-generalLossless} can be computed as follows. Let $X=I\{R_1=R_2\}$ indicate whether the two observations are the same or not; it is a Bernoulli random variable with success probability $Pr\{R_1=R_2\}$. The number of required experiments is a geometric random variable. The only possible error is to mistakenly declare type 4 as type 1, which happens with probability:
\begin{equation}
Pr(error)=1-Pr(type=1|X_1=1,\cdots,X_{countMax}=1)=1-\frac{1}{1+(Pr(X=1|type=4))^{countMax}}
\end{equation}
In type 4, $X=1$ occurs when $u\notin[W-D_1,W-D_2]$, \ie with probability $1-\frac{|D_1-D_2|}{(1-f)\cdot W}$. 
Thus, Alg.~\ref{alg-generalLossless} identifies any 2-by-2 topology in $countMax$ experiments with the following error probability:
\begin{equation}
\label{eq-error-Alg5}
Pr(error)=1-\frac{1}{1+(1-\frac{|D_1-D_2|}{(1-f)\cdot W})^{countMax}}
\end{equation}
We can then find $countMax$ by replacing the appropriate values \cite{active-topology-arxiv}. One can calculate that in order to ensure an accuracy of $99\%$ in distinguishing between types 1 and 4 2-by-2's, $countMax$ needs to be $\sim 450$. However, this is a pessimistic upper bound: simulation results in Section~\ref{sec-simul} show that a much smaller $countMax$ is sufficient in practice.

\subsubsection{Lossy 2-by-2} \label{sec-2by2Lossy}

Let us now consider a 2-by-2 network where packets may be lost on some links. In this
case, we can no longer guarantee meetings of $x_1$ and $x_2$ at the
joining points and predictable observations at the receivers. There
are two differences from the lossless case. First, because of random packet loss,
each experiment might result in different outcomes, shown in
Table~\ref{table-lossyObs}. Second, there are common observations across all four types,
as opposed to just between types 1 and 4. We divide the observations
in Table~\ref{table-lossyObs} into three groups: (i) at least one of
the receivers does not receive any packet (``-'') due to loss, (ii)
both receivers have the same observation $R_1=R_2$, and (iii) the
two receivers have different observations $R_1 \neq R_2$.

\begin{table*}[t!]
{\footnotesize \caption{\label{table-lossyObs}{\bf Lossy Case}.
Possible observations for all four types of 2-by-2 topologies.
(Sources send synchronized and $W$ is large. Observation \#13 for
types 2, 3 occurs only when $S_2$ sends with offset $u\in [f\cdot
W,W]$ after $S_1$.) We divide the observations into 3 groups:
(i) at least one receiver does not receive any packet, (ii) $R_1=R_2$ and (iii) $R_1 \neq R_2$.} \vspace{-1.5em}
\begin{center}
\begin{tabular}[t]{|c||l|c|c||l|c|c||l|c|c||l|c|c|}
\hline
  &   & \multicolumn{2}{|c||}{Type 1}  &  &  \multicolumn{2}{|c||}{Type 2} &  & \multicolumn{2}{|c||}{Type 3}  &  &  \multicolumn{2}{|c|}{Type 4} \\

$\#$ &     grp        &   $\mathbf{R_1}$ &  $\mathbf{R_2}$  &   grp    & $\mathbf{R_1}$   &  $\mathbf{R_2}$ &   grp   &   $\mathbf{R_1}$  & $\mathbf{R_2}$ &  grp    & $\mathbf{R_1}$  &  $\mathbf{R_2}$\\
\hline \hline
1 &   i  &       -        &    -             &  i  &      -            &          -       &  i &    -             &       -    &   i &    -      &    -    \\
\hline
2 &             &        -       &    $x_1+x_2$     &            &  -            &      $x_1+2x_2$  &             &  $x_1+2x_2$    &       -         &           &    -      &  $x_1+x_2$ \\
\hline
3 &             &          -     &    $x_1$         &            &    -           &      $x_1+x_2$     &           & $x_1+x_2$    &       -          &           &    -      &    $x_1$   \\
\hline
4 &             &     -          &    $x_2$          &           &   -           &         $x_1$       &           &  $x_1$     &       -         &              &    -      &    $x_2$  \\
\hline
5 &             &    $x_1+x_2$     &    -            &           & -            &        $x_2$        &            &  $x_2$      &       -          &            &    $x_1+x_2$   &    -    \\
\hline
6 &             &     $x_1$        &    -            &           & $x_1+x_2$     &          -          &           &   -         &     $x_1+x_2$    &            &     $x_1$      &    -    \\
\hline
7 &             &  $x_2$         &    -              &           & $x_1$        &          -          &            &    -         &      $x_1$       &           &  $x_2$      &    -    \\
\hline
8 &  ii  &   $x_1+x_2$    &    $x_1+x_2$       &          &  $x_2$       &          -          &            &   -         &      $x_2$       &   ii &    $x_1+x_2$   &   $x_1+x_2$ \\
\hline
9 &             &  $x_1$        &    $x_1$           &  ii & $x_1+x_2$    &       $x_1+x_2$     &  ii &  $x_1+x_2$   &    $x_1+x_2$     &           &   $x_1$     &    $x_1$  \\
\hline
10 &            &   $x_2$       &    $x_2$           &            & $x_1$        &       $x_1$         &           & $x_1$       &    $x_1$         &             &   $x_2$     &    $x_2$  \\
\hline
11 &            &               &                    &            &   $x_2$        &       $x_2$         &         &  $x_2$       &     $x_2$        &  \textbf{iii}     &    $\mathbf{x_1}$    &   $\mathbf{x_1+x_2}$ \\
\hline
12 &            &               &                    &  \textbf{iii} &    $\mathbf{x_1+x_2}$    &     $\mathbf{x_1+2x_2}$    &   \textbf{iii}  &  $\mathbf{x_1+2x_2}$    &     $\mathbf{x_1+x_2}$    &              &     $\mathbf{x_1+x_2}$  &    $\mathbf{x_1}$    \\
\hline
13 &            &               &                    &             &   $\mathbf{x_1}$   &      $\mathbf{x_1+x_2}$      &             &   $\mathbf{x_1+x_2}$   &   $\mathbf{x_1}$      &               &   $\mathbf{x_1}$   &   $\mathbf{x_2}$   \\
\hline
14 &            &               &                    &             &  $\mathbf{x_1}$       &      $\mathbf{x_2}$           &             &    $\mathbf{x_2}$    &      $\mathbf{x_1}$       &         &    $\mathbf{x_2}$      &      $\mathbf{x_1}$ \\
\hline
15 &            &               &                    &             &  $\mathbf{x_1+x_2}$     &     $\mathbf{x_2}$          &            &    $\mathbf{x_2}$      &   $\mathbf{x_1+x_2}$    &         &     $\mathbf{x_1+x_2}$    &    $\mathbf{x_2}$\\
\hline
16 &            &               &                    &             &            &                 &               &        &          &             &   $\mathbf{x_2}$    &   $\mathbf{x_1+x_2}$  \\
\hline
\end{tabular}
\end{center}
}
\vspace{-1.0em}
\end{table*}

We choose to ignore the observations of {\em group (i)} because they
can occur in any of the four 2-by-2 types and thus, they do not help to distinguish among 2-by-2's in the deterministic way adopted in
this paper. Observations of {\em group (ii)} can also be the result of any
2-by-2 type: unlike the lossless case, where $R_1=R_2$ is unique to
type 1 or 4 topologies, any of the four topologies may result
in such observations if some packets are lost. We observe that group (ii)
are the only possibility for type 1 topology, apart from the group
(i) that we ignore, while all other three 2-by-2 types may result in
either $R_1=R_2$ or $R_1 \neq R_2$. Therefore, if after $countMax$
trials, we only have observations of group (ii), we declare the topology
to be type 1.

\begin{algorithm}[t]
{\footnotesize \caption{\label{alg-generalLossy} \footnotesize {\bf Lossy Case -
Inferring a 2-by-2 component.} Sources $S_1$ and $S_2$ multicast
$x_1$ and $x_2$, respectively. Receivers observe $R_1=c_{11}x_1+c_{12}x_2$ and
$R_2=c_{21}x_1+c_{22}x_2$. The variable \textit{type} stores our estimate of
the type of the 2-by-2 component and it gets updated during the experiments.}
\begin{algorithmic}[1]

\STATE $n=1$; /*first experiment*/ \STATE type=0; /*initialization*/
\WHILE{$n\leq countMax$} \IF{$R_1\neq[0,0] ~\& ~ R_2\neq [0,0]$}
\IF{$c_{22}>c_{12}$} \IF{type $\neq 3$} \STATE type=2; \ELSE \STATE
type=4; Break; \ENDIF \ELSIF{$c_{22}<c_{12}$} \IF{type $\neq 2$}
\STATE type=3; \ELSE \STATE type=4; Break; \ENDIF \ELSIF{$type==0
~\& ~ R_1==R_2$} \STATE type=1; \ENDIF \ENDIF \STATE n++; \STATE
Draw offset $u$ uniformly at random out of $[f\cdot W,W]$.

\STATE Send probes; $S_2$ transmits $u$ time later than $S_1$.
\ENDWHILE \STATE Output type.
\end{algorithmic}
}
\end{algorithm}

In observations of {\em group (iii)}, it is $R_1 \neq R_2$, which
means that $c_{12}\neq c_{22}$ and/or $c_{11}\neq c_{21}$. An important observation
is that the difference of the coefficients between the two receivers contains
topology-related information. W.l.o.g., we focus on the coefficient
of $x_2$ and look at the difference $c_{12}-c_{22}$.
Table~\ref{table-lossyObs} shows that $c_{12}-c_{22}<0$ can only
occur in type 2 or type 4 topologies; while $c_{12}-c_{22}>0$ can
only occur in a type 3 or 4 topology. Note that the coefficient is larger on one side (\eg $c_{12}>c_{22}$) when the probe ($x_2$) goes through two joining points on its way to one receiver (in this case, $R_1$) and through one joining point on its way to the other receiver ($R_2$). By performing several independent experiments and collecting several observations of group
(iii), we can distinguish among the candidate topologies. If after
$countMax$ experiments, there are only observations of group (ii) or
(iii) with $c_{12}-c_{22}\leq0$, we declare the topology as type 2. If
there are only observations of group (ii) or (iii) with
$c_{12}-c_{22}\geq0$, we declare it as type 3. If there are
observations of group (ii) or (iii) with both $c_{12}-c_{22}<0$ and
$c_{12}-c_{22}>0$, we declare it as type 4.

In our experiments, we try to create those observations that reveal the
topology. These can occur either naturally, as the result of packet
loss, or artificially, by us introducing an offset $u$ in $S_2$'s
sending time with respect to $S_1$. To help these observations
occur, especially for small loss rates, and similarly to the lossless
case, we use a random offset $u\in [f\cdot W,W]$.
To make these experiments independent, we space apart successive sets of probes by
roughly selecting $T \ge 3W$, which is sufficient since there are at most two joining points on any 
$(S_i,R_j)$ path in a 2-by-2.

Alg.~\ref{alg-generalLossy} summarizes the 2-by-2 inference
for lossy networks. The algorithm is simple and follows a deterministic
approach: one observation, or a set of observations, is sufficient
to uniquely distinguish among types. For example, at least
one observation of group (iii) rules out the type 1 topology; a
pair of group (iii) observations with both $c_{12}-c_{22}>0$ and
$c_{12}-c_{22}<0$ indicates type 4; etc. As a result, we require
less experiments compared to thousands of arrival order
measurements required by \cite{journal, probing} for statistical
significance. In addition and more importantly, we identify the
exact 2-by-2 type, while \cite{journal} was only able to distinguish between
shared and non-shared types. The following Lemma describes the probability of error of Alg.~\ref{alg-generalLossy} with respect to the number of experiments ($countMax$) more precisely.

\begin{lemma} \label{thm-Thm3}
Alg.~\ref{alg-generalLossy} identifies any 2-by-2 topology with $Pr(error)\leq \sum_{i=1}^{2} (1-\rho'_i)^{countMax}$ in $countMax$ experiments, where $\rho'_i={(1-p)}^6 (p+(1-p) \gamma_i) (p+(1-p) \overline \gamma_{i'})$, $i,i'=1,2$, $i\neq i'$, $p$ is the link loss rate (same for all links), and $\gamma_i$ is the probability that probe packet $x_2$ arrives within $W$ at $J_i$ in a type 4 topology, \ie $\gamma_i=Pr(u+D_i\leq W)$, $i=1,2$.
\end{lemma}
The proof is provided in Appendix B.

\subsubsection{Inferring all 2-by-2's in a 2-by-N Network}
\label{sec-2byN}

Algorithms~\ref{alg-generalLossless} and~\ref{alg-generalLossy} can
be directly applied to a 2-by-N network, where
two sources multicast to $N$ receivers. A difference is that
intermediate nodes need to perform addition over a larger finite
field, of order larger than the maximum number of joining points on a path (g), since a packet may meet itself at all the joining points on the path. In the worst case, there can be $N$ joining points in a row and thus, the maximum required field size is the first prime greater than N. Algorithm \ref{alg-generalLossless} and Algorithm~\ref{alg-generalLossy} can
be performed on any pair of receivers among all $\binom{N}{2}$
possible pairs. The same set of 2-by-N probes can be used to infer, in
parallel and independently, the type of all 2-by-2 topologies. This
reduces the number of probes, as we re-use them, instead of sending
$\binom{N}{2}$ different sets of probes. The 2-by-N structure is
important for the merging algorithm in Section~\ref{sec-merging}.

\subsubsection{2-by-2's vs. other Subnetwork Components} \label{sec-otherStructures}

We now discuss why we choose to decompose an M-by-N network into 2-by-2 subnetwork components, as opposed to any other subnetwork structures $m-by-n$, $1\le m\le M, 1\le n \le N$:
\begin{itemize}\addtolength{\itemsep}{-0.35\baselineskip}  
\item 1-by-1: This is the smallest component and corresponds to measuring a single end-to-end path.
However, it reveals neither joining nor branching points.
\item 1-by-2 and 2-by-1: These correspond to a 2-leaf multicast or a reverse-multicast tree,
respectively. The 2-by-1 consists of 2 sources, one coding point, and 1 receiver.
The 2-by-1 cannot identify the branching points while the 1-by-2 cannot identify the joining points. Similar comments apply to M-by-1 and 1-by-N.
\item 2-by-2: This is the smallest structure that gives information about the relative
locations of joining and branching points.
\item m-by-n, with $2<m<M, 2<n<N$: If we consider larger structures,
there is an exponentially larger number of possible types,
which requires more complicated inference algorithms. For example, there exist 19 possible types for a 2-by-3 structure.
\item M-by-N: In the extreme case, we need to enumerate all possible M-by-N topologies as in \cite{jaggi}.
\end{itemize} 
The larger the subnetwork component we use as a building block, the less components we need to infer and the simpler the merging algorithm. However, as the size of the basic component grows, the number of possible types increases exponentially and the inference step becomes increasingly complex. In this paper, we choose to decompose an M-by-N into 2-by-2 components, inspired by the approach in \cite{journal}. We note that 2-by-2 is the minimum size building block required to infer both joining and branching points and strikes a good tradeoff of inference vs. merging complexity.

\subsection{Merging Algorithm} \label{sec-merging}

Assuming knowledge of all 2-by-2 subnetwork components, from Section~\ref{sec-2by2},
we now merge them together to reconstruct the M-by-N network. We study merging in two different
scenarios: (i) when a 1-by-N tree topology is known, which is the same problem
studied in \cite{journal}; and (ii) without knowledge of any 1-by-N,
which is new to our work. Exploiting the accurately identified
2-by-2's, we can solve (i) exactly, which was previously only
approximately solved; and also solve (ii), which was previously not known how to address.

More precisely, our merging algorithm can identify every joining
point, in the sense that it can localize it between two branching points.
However, note that when there are several joining points in a row, without any branching point in between, it is not possible to identify the relative locations of these joining points with respect to each other. In fact, this is the case in a tree topology.

\subsubsection{\label{sec-merging1} Merging a 1-by-N and 2-by-2's into a 2-by-N}

In this section, we assume that the 1-by-N from $S_1$ to $N$
receivers is known using either the classic methods for single-tree topology inference \cite{survey} or our algorithms in Section~\ref{sec-trees} for tree networks.

This 1-by-N is a tree rooted at $S_1$ and contains
only branching points. We also assume that the 2-by-2's between
$S_1$, a new source $S_2$, and any pair of receivers are known,
using the algorithms of Section~\ref{sec-2by2}. Our goal is to
locate the joining points where paths from $S_2$ to the same $N$
receivers join $S_1$'s topology. We use the assumptions of
Section~\ref{sec-statement} for routing.

This problem was posed in \cite{journal, merging} and was solved there
in an approximate way. Bounds on the joining point locations in the
$S_1$ topology were provided within a sequence of consecutive
logical links. %or within a single logical link.
This was because the 2-by-2's are only identified as
shared or non-shared types in \cite{journal, probing}.

In contrast, we design Algorithm~\ref{alg-merging1}, which localizes each joining point for
each receiver to a single logical link, between two branching
points in the $S_1$ topology. Our algorithm is simpler, faster, and
more accurate: it can identify {\em all} joining points for any
topology and with lower complexity, thanks to our complete knowledge
of the 2-by-2 types.

\begin{example}
 Fig.~\ref{fig-abileneTopology} depicts a 2-by-9 topology constructed based on the Abilene network \cite{abilene}. Consider $R_1$: it forms a type 1 2-by-2 with $R_2$. Therefore, $J_1$ must
lie above $B_{1,2}$, so that there exists a unique path from each
source to $R_1$. We then need to localize $J_1$ with respect to $B_{1,3}$: $R_1$,$R_3$ form
a 2-by-2 of type 4; thus, $J_1$ must lie below $B_{1,3}$. $J_1$ is now localized to one link (between $B_{1,2}$ and $B_{1,3}$), and the algorithm ends here for $R_1$. Other
receivers are considered similarly. Note that a joining point can be
placed on any link from the receiver to $S_1$. Therefore, the number of steps
required to localize a joining point is at most equal to the height of the $S_1$ tree.
Also, when there is a group of receivers within which all
pairs are of type 1, the algorithm is run only once and it assigns the same joining
point to all of them. For this example, the
algorithm in \cite{journal} cannot completely resolve all joining
points, and provides bounds within a sequence of several logical
links instead.
\end{example}

\begin{algorithm}[t!]
{\footnotesize \caption{\label{alg-merging1} \footnotesize {\bf Merging
Algorithm:} Given the two sources $S_1,S_2$, a set of
receivers $R_1,R_2,\cdots,R_N$, the 1-by-N $S_1$ tree topology, and the
2-by-2 results from Alg.~\ref{alg-generalLossy} for any pair of
receivers $R_i,R_j$, this algorithm identifies a single link for the
location of every $J_i$ (the joining point for $R_i$), on $S_1$
topology.}
\begin{algorithmic}[1]

\FOR{each receiver $R_i$} \IF{$\exists$ $k<i$ such that the
$S_1,S_2,R_k,R_i$ 2-by-2 is shared} \STATE $J_i=J_k$; \ELSE \STATE
Let $B$ be the closest branching point to $R_i$. \WHILE{$J_i$ is not
localized to a single link} \STATE Let $R_j$ be any child of $B$
($j\neq i$). \STATE Based on the type of the 2-by-2 component
$S_1,S_2,R_i,R_j$, locate $J_i$ above/below $B$. \IF{($J_i$ is below
$B$) $||$ (($J_i$ is above $B$) $\&\&$ ($\nexists$ other branching
point above $B$ on $S_1$'s 1-by-N))} \STATE $J_i$ is localized to a
single link. \STATE Output this link; Break; \ELSE \STATE
$B=$ the next upstream branching point. \ENDIF \ENDWHILE \ENDIF
\ENDFOR

\end{algorithmic}
}
\end{algorithm}

\subsubsection{Merging 2-by-2's into a 2-by-N} \label{sec-merging2}

In this section, we infer a 2-by-N  without prior knowledge of any 1-by-N. Inference under this relaxed assumption is enabled by our exact knowledge of 2-by-2's and was not possible before \cite{journal, merging}. We first send probes over the 2-by-N and then merge all $\binom{N}{2}$ 2-by-2's, as described next.

\begin{example}
We first consider all shared (type 1) 2-by-2 components
and assign them the minimum number of branching and joining points
required. For example in Fig. \ref{fig-abileneTopology}, $B_{1,2},
B_{3,4}$ and $J_1=J_2, J_3=J_4=J_5=J_6=J_7=J_8=J_9$ are identified
in this step. Second, we consider all non-shared 2-by-2
topologies (of type 2, 3, or 4). We use the information about the
locations of the branching and joining points in each type to: (i)
add the minimum number of branching points required to the ones
already identified from the shared pairs; and (ii) assign joining
points to those receivers that have not been already assigned one.
In the example of Fig.~\ref{fig-abileneTopology}, an additional
branching point $B_{1,3}$ is required, which is connected to both
joining points $J_1=J_2$ and $J_3=J_4=J_5=J_6=J_7=J_8=J_9$, to
satisfy the 2-by-2's of type 4 between the two shared groups. No
additional joining point is required in this example.
\end{example}

This approach identifies the locations of all joining points,
between the $S_1$ and $S_2$ 1-by-N topologies, but it does not identify
all the branching points in the $S_1$ tree topology. Only the
``minimum'' $S_1$ topology is identified, {\em i.e.,} the tree made
by the ``necessary'' branching points. We define as
``necessary'' branching points the ones located below a joining
point of $S_1$ and $S_2$ in the 2-by-N. An ``unnecessary'' branching point is the child of another branching point with no joining point in between. For example in Fig. \ref{fig-abileneTopology}, this approach does not identify $B_{4,5}, B_{6,7}, B_{6,9}$, and directly connects their
children ($R_4, R_5, R_6, R_7, R_8, R_9$) to the upstream branching
point ($B_{3,4}$).

Note that the worst case input for this approach is a tree network. Since all
2-by-2's are of type 1, and the algorithm cannot reconstruct branching points
in a row, it can only identify the top-most branching point of the entire tree structure.

\subsubsection{From 2-by-N to M-by-N}

We can directly extend the 2-by-N inference techniques to the M-by-N
case \cite{merging}. We start from a 2-by-N topology, and add one
source at a time, to connect the 1-by-N's of the remaining $M-2$ sources. Assume
that we have constructed a k-by-N topology, $2 \leq k < M$. To add
the $(k+1)^{th}$ source, we perform $k$ experiments, where at each
experiment one different of the $k$ sources and the $(k+1)^{th}$
source send $x_1$ and $x_2$. We then {\em glue} these topologies together
by following the topological rules of Section~\ref{sec-merging1} (with single-source
trees given) or Section~\ref{sec-merging2} (without that assumption).

\subsubsection{Complexity of Merging}

\begin{lemma}
\label{theorem-minimum2by2s}
If one source's 1-by-N tree topology is given, the minimum number of 2-by-2's required by any merging algorithm to uniquely localize all the joining points (between two branching points) in the 2-by-N topology is $\frac{N}{2}$.
\end{lemma}
\begin{proof}
One can think of checking the types of the 2-by-2 components in the following sense: we divide the $N$ receivers in the network into two sets of vertices, in a bipartite graph, and we draw an edge between any two vertices for which we check the 2-by-2 type. The minimum number of required 2-by-2's is then given by a perfect matching in this bipartite graph; therefore, it is $\frac{N}{2}$.
\end{proof}

\begin{example}
Fig.~\ref{fig-minimumExamples} shows two 2-by-N topologies that require exactly $\frac{N}{2}$ 2-by-2's for their joining points to be uniquely identified by any merging algorithm. In Fig.~\ref{fig-minimumExamples}(a), checking the types of  $(R_1,R_2)$ and $(R_3,R_4)$ is sufficient for localizing all four joining points. In Fig.~\ref{fig-minimumExamples}(b), where all the joining points are the same as $J$, checking the types of  $(R_1,R_3)$ and $(R_2,R_4)$ would be sufficient.
\end{example}

\begin{figure*}[t!]
\centering \subfigure[Example topology 1]{\includegraphics[width=3.5cm,height=4cm]{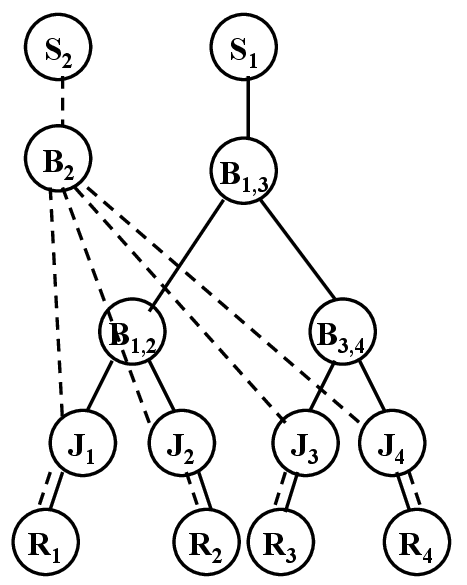}} \hspace{1cm} \subfigure[Example topology 2]{\includegraphics[width=3.5cm,height=4cm]{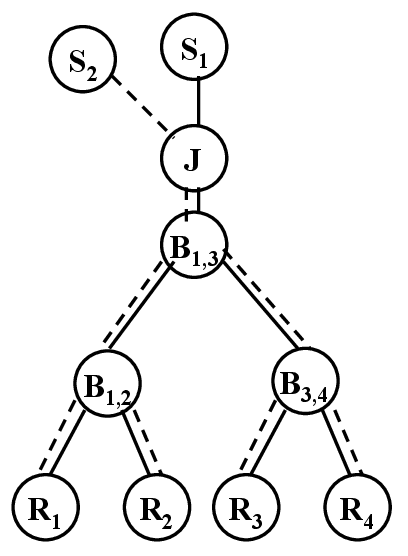}} \caption{Two example 2-by-N topologies that require exactly $\frac{N}{2}$ 2-by-2's for their joining points to be uniquely identified by any merging algorithm.}
\label{fig-minimumExamples}
\end{figure*}

{\em Note on Lemma \ref{theorem-minimum2by2s}:} If the 2-by-2's are properly selected, $\frac{N}{2}$ of them can be sufficient in some topologies, as we see in the examples of Fig.~\ref{fig-minimumExamples}. Unfortunately, we do not know in advance (without knowledge of the 2-by-N topology) which 2-by-2's to choose out of all $\binom{N}{2}$ possible 2-by-2's, so as to uniquely localize the joining points between branching points. Nevertheless, from the given $S_1$ 1-by-N topology, we can give an upper bound on the number of 2-by-2's required. Since every receiver is checked with other receivers that are children of its upper branching points, up to the location of its joining point, we need to check for $O(N\log N)$ 2-by-2's. This is less than identifying all $\binom{N}{2}$ 2-by-2's. Note that we still need to multicast $x_1,x_2$ to all receivers and monitor all observations, but we can use only the observations of the selected 2-by-2's for inference, and ignore the rest.

\begin{lemma}
Algorithm~\ref{alg-merging1} takes at most $O(N\log N)$ steps.
\end{lemma}
\begin{proof}
As mentioned in the note above, Algorithm~\ref{alg-merging1} considers every single receiver and checks the 2-by-2 type of that receiver with other receivers that are children of its upper branching points, up to the location of its joining point. Therefore, it takes at most $O(N\log N)$ steps.
\end{proof}

This is an improvement over $O(N^3)$ in \cite{journal}. Note that the second merging algorithm requires all $\binom{N}{2}$ 2-by-2's; therefore, it takes $\binom{N}{2}$ steps.

\section{\label{sec-simul}Simulation Results}

We now simulate our Inference\footnote{We note that, in both our approach and in past work
\cite{journal,probing}, the error in identifying the 2-by-2's, in the first step, may propagate to the Merging algorithm, in the next step. However, there is no additional error introduced by the Merging
algorithm itself, and thus no need to simulate it.} algorithms in some representative topologies that exemplify different characteristics.

\subsection{Trees}

\subsubsection{Simulation Setup}
Consider the binary tree example of Fig.~\ref{fig-binaryExample}. Assume that all links have the same loss probability $p\in[0,10\%]$. We simulate Algorithm~\ref{alg-binaryLossy} and we send up to $M=10$ probes per iteration. We conservatively consider an error to be {\em any} divergence from the true topology. The results are averaged over $10,000$ realizations of the loss process.

\subsubsection{Simulation Results}
Fig.~\ref{fig-binaryTreeSimul} shows the percentage of inference
errors in each of the first two iterations (shown in Fig.~\ref{fig-binaryExample2} and Fig.~\ref{fig-binaryExample3}) as a function of $p$ and $M$. As expected, the probability of error is increasing with $p$, since packet losses may lead to the misclassification of a leaf to
the incorrect component. For a fixed number of probes per iteration and fixed loss rate $p$, the probability of error decreases with the iterations.
%, as the size of the inferred network also decreases.
It also decreases rapidly with the number of probes $M$ per iteration: it decreases significantly even with $M=2,3$, even for large $p$, and becomes practically zero for $M\ge5$.\footnote{This second observation is due to the fact that {\em one} correctly received packet is sufficient for the
correct operation of Alg.~\ref{alg-binaryLossy}. {\em E.g.,} if a
node receives a mixture of $x_1$ and $x_2$ probes, it will be correctly
assigned to component $\mathcal{L}_3$ even if some probes are
lost. In contrast, methods that require each receiver to receive enough packets to infer the
loss rate associated with the network links with a certain accuracy, need a larger number of probes for statistical significance.}

\begin{figure}[t!]
\centering \subfigure[Iteration 1 infers the topology in Fig.~\ref{fig-binaryExample2}.]{\includegraphics[width=7.5cm]{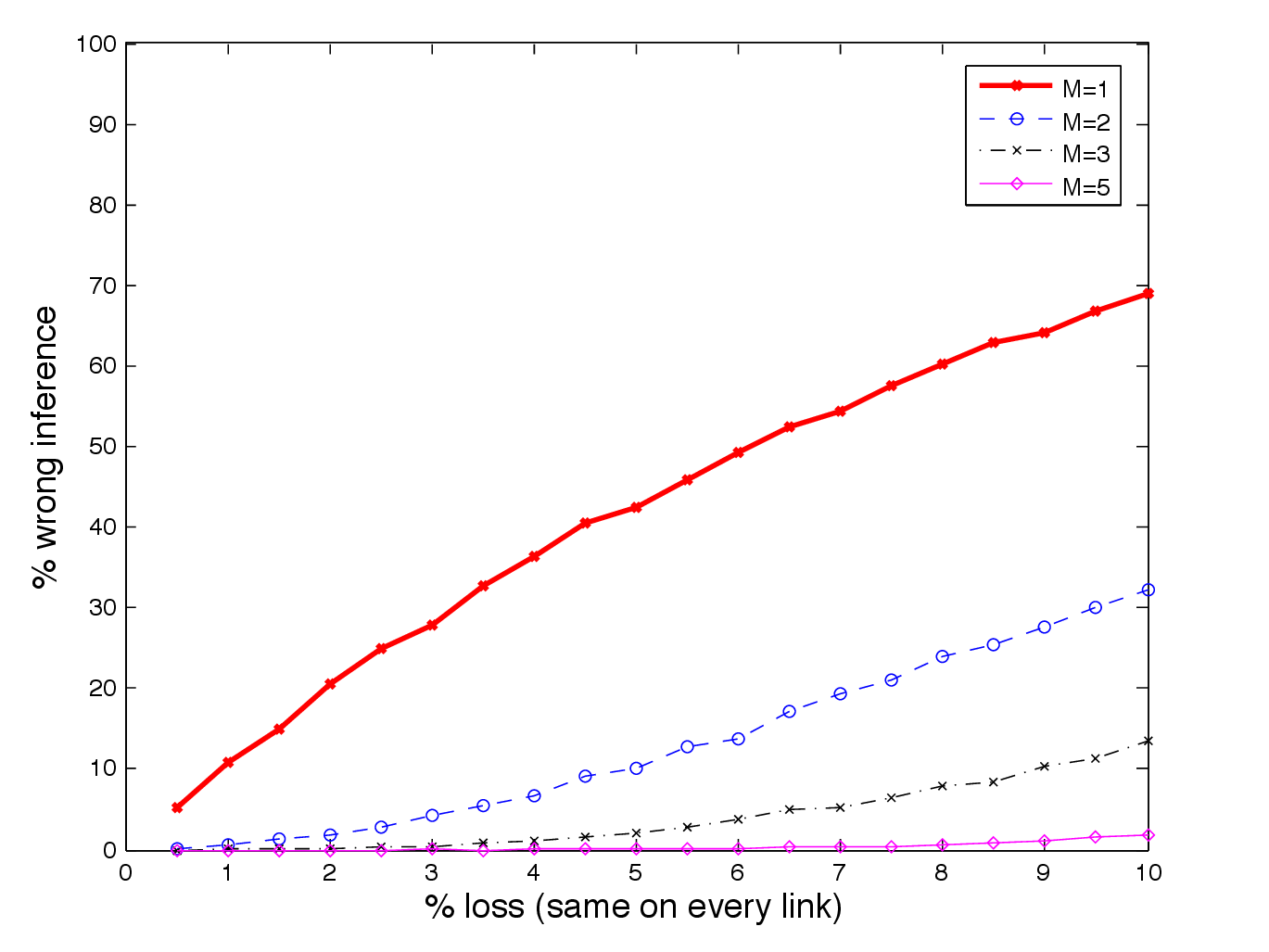}} \subfigure[Iteration 2 infers the topology in Fig.~\ref{fig-binaryExample3}.]{\includegraphics[width=7.5cm]{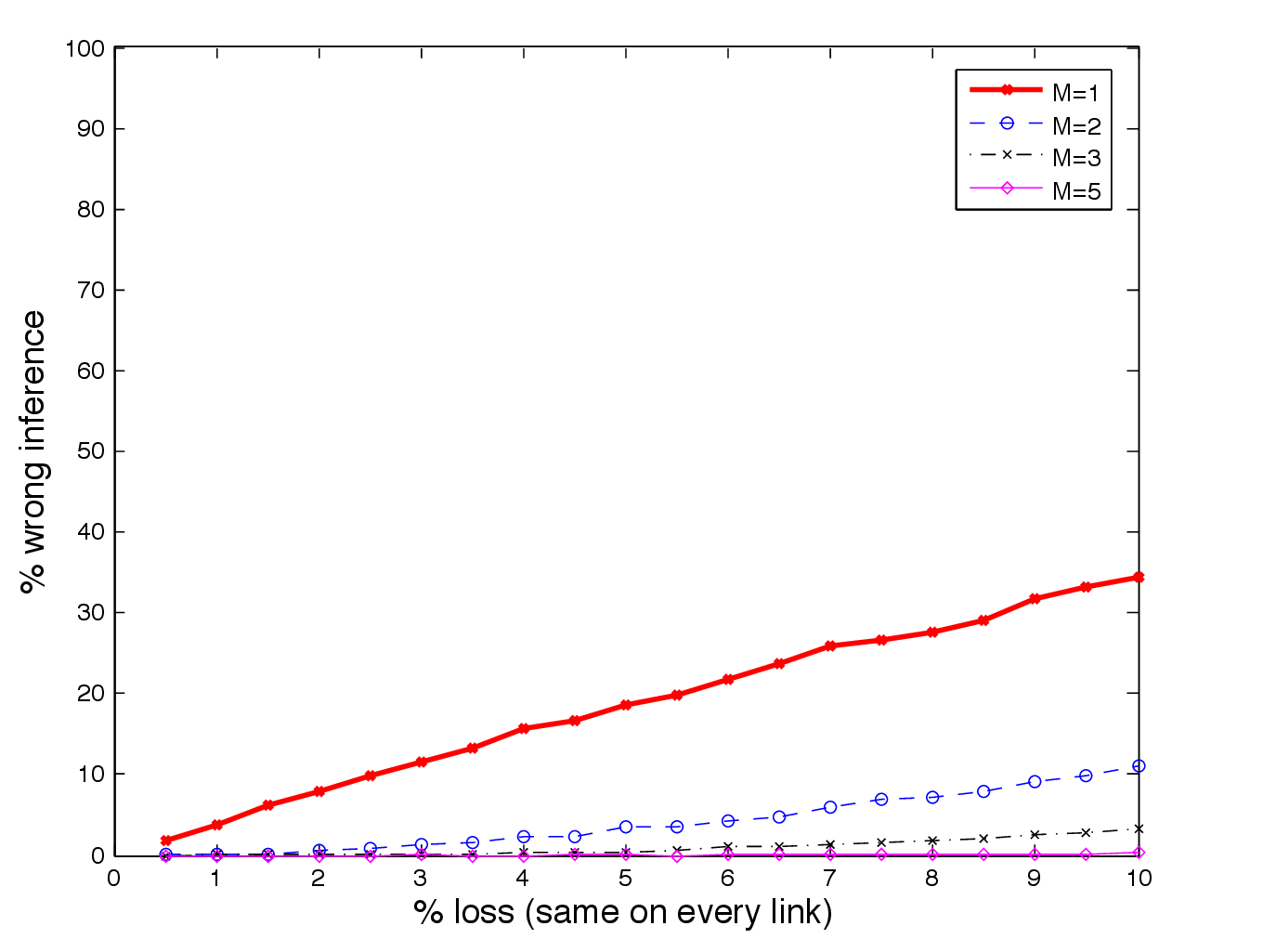}} \caption{\label{fig-binaryTreeSimul} Probability of incorrect inference for the binary tree of Fig.~\ref{fig-binaryExample}, as a function of the loss probability $p$ (same for all links) and of
the number of probes $M$ per iteration. The results are averaged over $10,000$ realizations.}
\vspace{-0.5em}
\end{figure}

\subsection{Multiple-Tree Topologies (DAGs)}
We simulate our algorithms in example multiple-tree networks. In summary, we show that (i) our approach significantly improves over \cite{probing, merging, journal}, in terms of the number of experiments required to identify the type of all 2-by-2's as well as of the associated probability of error;  (ii) the probability of error in identifying the 2-by-2's depends on the underlying topology. In particular, it is smaller for preferential attachment graphs as compared to ER random graphs.
%random graphs that may represent wireless multi-hop topologies.

\subsubsection{Simulation Setup}
\label{sec-generalSimulSetup}
 To demonstrate (i), we consider Fig.~\ref{fig-threeTopologies}(a), which shows the Abilene topology \cite{abilene}, with two sources located at the Chicago and Indiana nodes, and nine
receivers, each located at one of the other core network nodes. This is the same topology considered in \cite{journal}. To investigate (ii), we consider Fig.~\ref{fig-threeTopologies}(b) and Fig.~\ref{fig-threeTopologies}(c). Fig.~\ref{fig-threeTopologies}(b) shows a random topology with 2 sources and 7 receivers generated by LEDA \cite{leda}, which can be used to model wireless multi-hop topologies.
Fig.~\ref{fig-threeTopologies}(c) shows a preferential attachment topology generated by Brite \cite{brite}. We pick 2 sources and 8 receivers and we select the route for every source-destination pair, according to our assumptions in Section \ref{sec-statement}. 

We run Alg.~\ref{alg-generalLossless} and Alg.~\ref{alg-generalLossy} in the absence and presence of packet loss, respectively, and we compute the error. In the lossless case, we identify the 2-by-2 types and we report the error as a function of the number of experiments $countMax$. The only possible error in this case is to falsely declare type 4 as type 1. In the lossy case, we also report the error assuming that there is packet loss in the network (with probability $p$ independently on every link), and after applying Alg.~\ref{alg-generalLossy} to each topology. An error in this case can result either from declaring type 2 or 3 or 4 as type 1; or from declaring type 4 as type 2 or 3. We consider values of $p \in [0,10\%]$ and $countMax=100,200,250$.

We assume that individual link delays have a fixed part of 5-10ms (propagation delay) and a variable part, which is exponential with mean 2ms (queueing delay). We choose a large time window $W=100ms$. The offset $u$ is drawn uniformly at random from $[70,100]ms$, \ie $f=0.7$.

\subsubsection{Simulation Results}

Fig. \ref{fig-generalSimulLossless} reports the results for the lossless case for all three topologies in Fig. \ref{fig-threeTopologies}. We also report the results for the lossy case: only one topology is shown in Fig.~\ref{fig-generalSimul}(b) due to the lack of space; additional figures can be found in the technical report \cite{active-topology-arxiv}.

\begin{figure*}[t!]
\centering \subfigure[The Abilene topology \cite{abilene}.]{\includegraphics[scale=0.28]{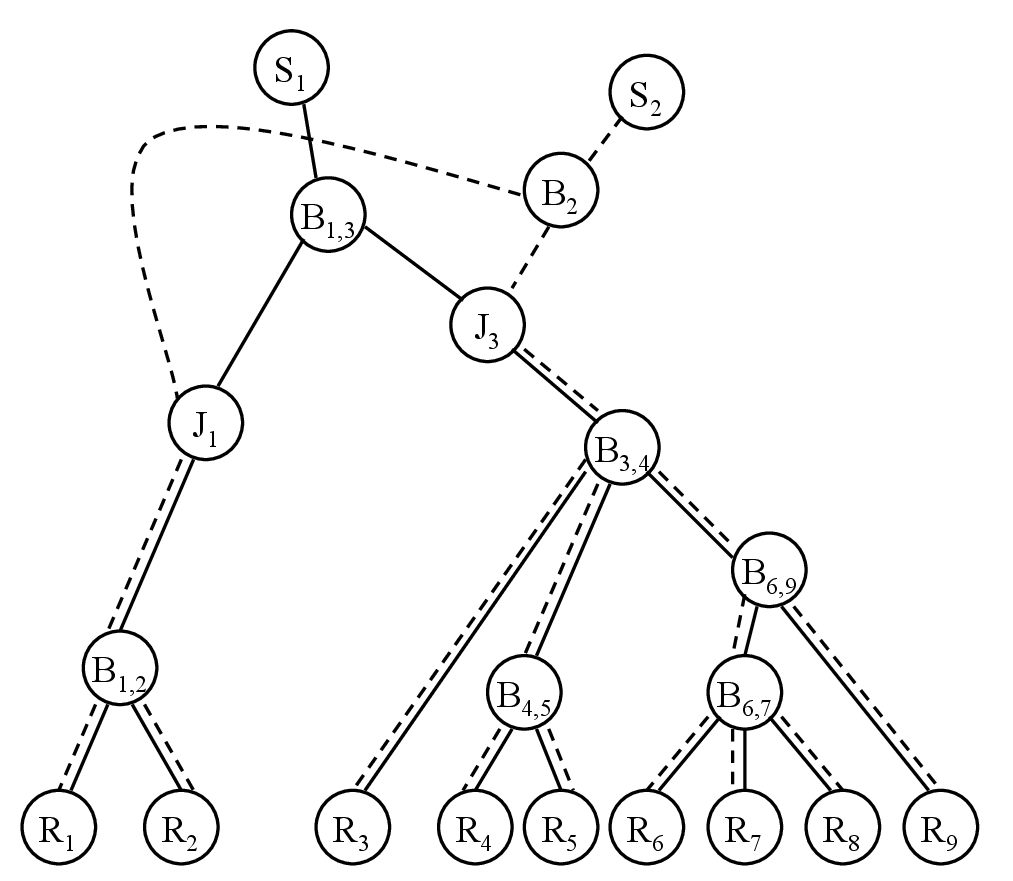}\label{fig-abileneTopology}} \hspace{0.3cm} \subfigure[An Erdos-Renyi random graph.]{\includegraphics[scale=0.31]{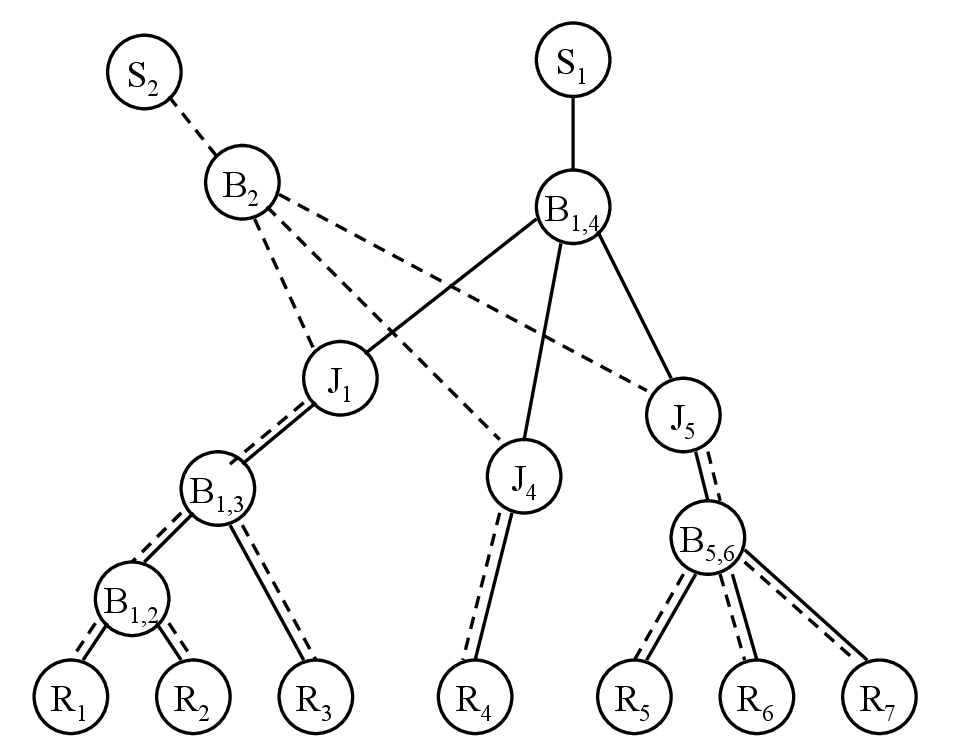}} \hspace{0.3cm} \subfigure[A preferential attachment graph.]{\includegraphics[scale=0.3]{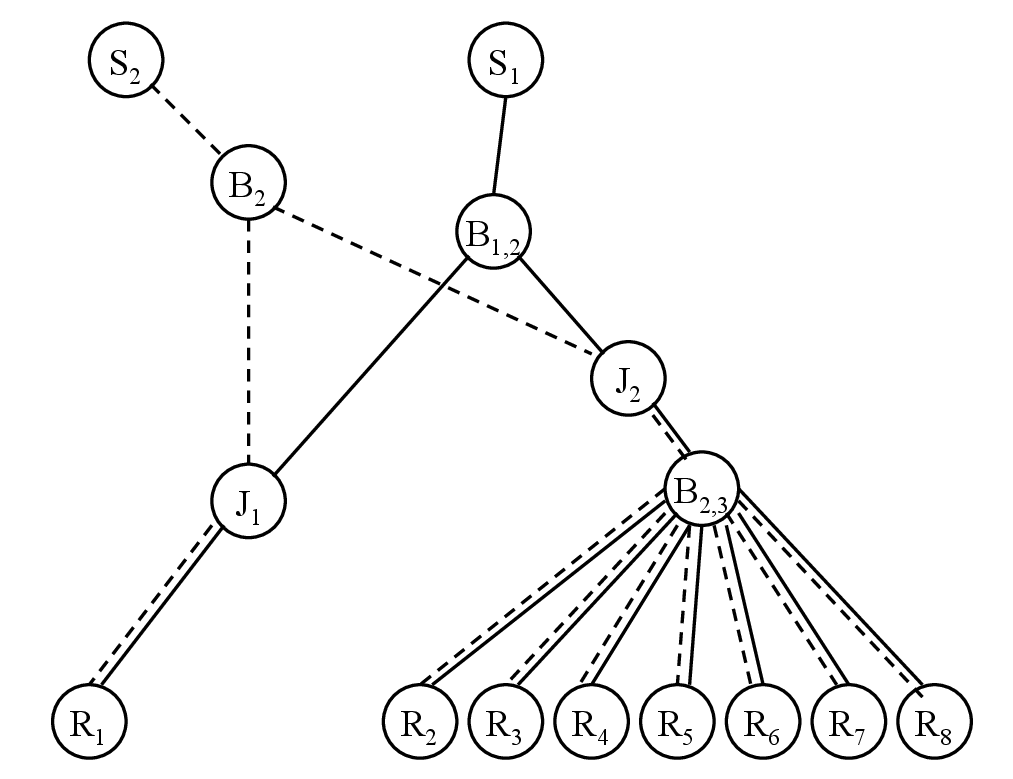}} \caption{\label{fig-threeTopologies} Three different topologies used to test our inference algorithms in simulation.}
\end{figure*}

\begin{figure}[t]
\centering \subfigure[Lossless case. Probability of error vs. the number of experiments
for the three topologies in Fig.~\ref{fig-threeTopologies}. The results are averaged over 1000 realizations.]{\includegraphics[width=7.6cm]{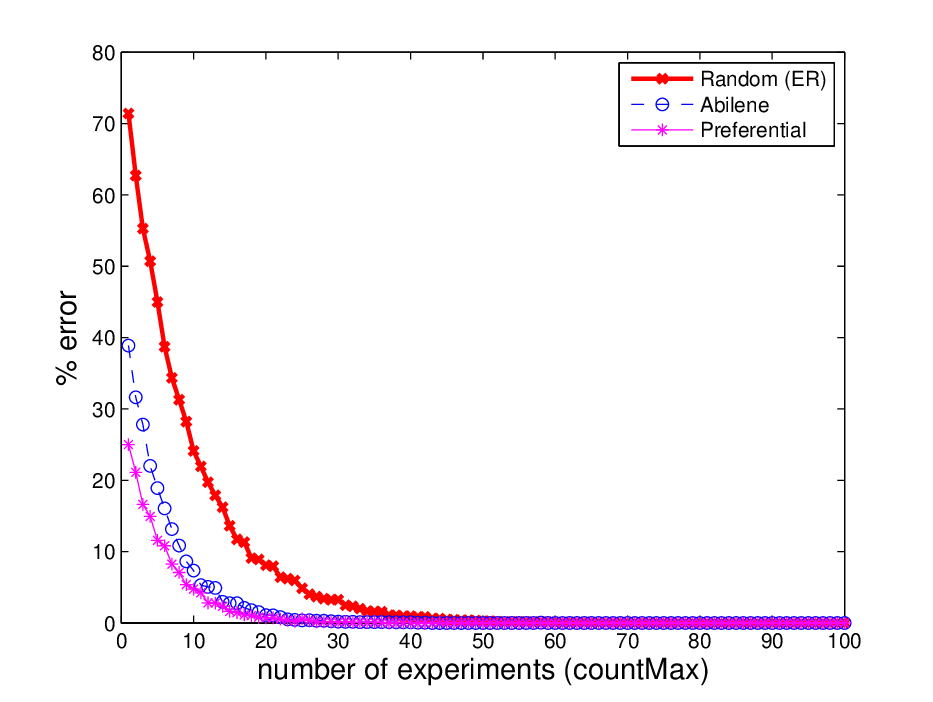}\label{fig-generalSimulLossless}} \hspace{0.5cm} \subfigure[Lossy case. The probability of error vs. the loss rate for different $countMax$ values for the Abilene topology in Fig.~\ref{fig-threeTopologies}(a); 5000 realizations.]{\includegraphics[width=7.5cm]{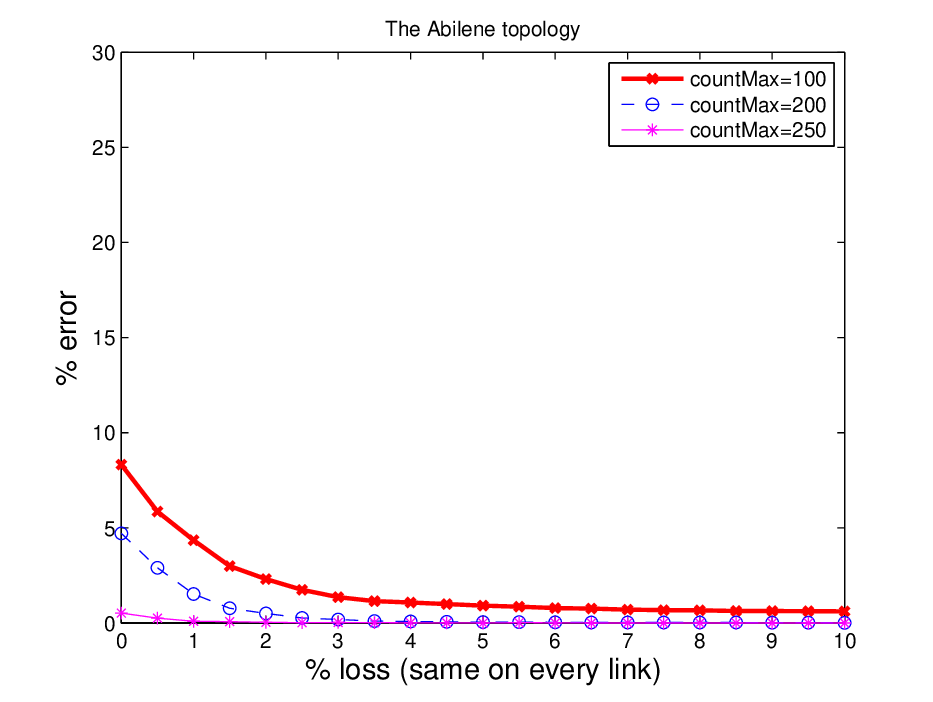}\label{fig-lossyAbilene}} \caption{\label{fig-generalSimul} Simulation results both in the absence and in the presence of loss for the topologies in Fig.~\ref{fig-threeTopologies}.}
\end{figure}

We first discuss the {\em Abilene topology} in Fig. \ref{fig-threeTopologies}(a). In the lossless case, the Abilene curve in Fig. \ref{fig-generalSimulLossless} shows that the error probability decreases very rapidly with $countmax$ and reaches 0 at $countMax\simeq50$. In the lossy case shown in Fig.~\ref{fig-lossyAbilene}, the error probability also decreases rapidly with $countMax$; it is negligible with $200-250$ experiments. This is a significant improvement over \cite{journal} for the same example topology: they use 1000 measurements to distinguish only between type 1 and the other types, for very small loss rates up to $1.5\%$, and they achieve error probability 10$\%$. In contrast, with an order of magnitude less probes, we distinguish among all four types, and we have a very small error probability for larger loss rates (up to 10\%). Note that the error probability is decreasing with the loss rate because loss actually helps to create observations of group (iii) \cite{active-topology-arxiv}.

We now consider {\em random graphs}, in particular Erdos-Renyi (ER) vs. Preferential Attachment, as shown in Fig.~\ref{fig-threeTopologies}(b), \ref{fig-threeTopologies}(c). In the lossless case, Fig.~\ref{fig-generalSimulLossless} shows that the error probability decreases rapidly with $countMax$ and reaches 0 at $countMax\simeq50$. We also observe that for the same number of experiments, the error is generally smaller  in the topology generated using the preferential attachment rather than the ER model. This is true in both lossless and lossy cases \cite{active-topology-arxiv}.\footnote{The reason is that in the preferential attachment topologies, we have a few nodes with a very high degree and many nodes with a low degree. As a result, we have a large number of receivers with a shared joining point and some other receivers with distinct joining points. In contrast, in ER graphs, we have several roughly equally-sized groups of shared receivers, where each group forms a non-shared type with any other group. Therefore, we have more 2-by-2's of type 1 in preferential attachment graphs, which results in a smaller error in Alg.~\ref{alg-generalLossless} and Alg.~\ref{alg-generalLossy}.}

\section{Our Work in Perspective} \label{sec-connections}

First, we discuss possible deployment scenarios in Section \ref{sec-deployment}. Then, in Sections \ref{sec-traditional}-\ref{sec-traceback}, we revisit the related work, briefly outlined in Section \ref{sec-related}, but now focusing on the most closely related parts. We discuss the trade-offs involved in each approach and we identify  connections and differences between our approach and each of the alternative approaches.

\subsection{Deployment Scenarios}
\label{sec-deployment}

In practice, our scheme can be deployed in two ways: 
\begin{itemize}\addtolength{\itemsep}{-.35\baselineskip} 
\item {\em Active probing scenario:} special probes are sent by the sources, recognized and coded/multicast at intermediate nodes, received by the receivers, sent and processed at a fusion center. 
\item {\em Packet marking scenario:} a special header is reserved on regular packets, marking and coding occurs only  on this header, without coding the data in the packet. The rest of the ideas in this paper apply in the same way, but now considering the content of the headers as opposed to the content of the payload.
\end{itemize}
The first scenario is the one already described throughout the paper. The sources send special probe packets, whose sole purpose is to be received and processed for inference purposes. The intermediate nodes perform simple network coding operations on these packets, to allow for more accurate and efficient inference. We are going to compare the first scenario to alternative topology inference approaches in Sections~\ref{sec-traditional}, \ref{sec-passive}, \ref{sec-randomNC} and \ref{sec-traceroute}. In the second scenario, our scheme could also be implemented as passive on top of regular traffic. A 2-by-2 component consists of  four regular unicast flows between the two sources and the two receivers,  intermediate nodes mark a special field in the headers of these regular packets. The special field in the headers can be coded (by constructing a mark that is the sum of the marks  from two flows going through the intermediate node), or multicast (by marking several outgoing flows with the mark found on the header of one incoming flow).  However, all the ``action'' is now on the headers, while the actual data can remain uncoded. In this respect, this deployment scenario is similar to packet marking techniques, which we describe in detail in Section~\ref{sec-traceback}.

Beyond the deployment scenario and practical applicability, we believe that our work provides a fundamental building block for exploiting correlation in the content of network coded packets for inference. It is already applicable to networks with network coding (which can be applied to dedicated probes sent over the network, as in the first scenario above) or packet marking capabilities (which can be applied only on the headers of packets from regular unicast flows, as in the second scenario above), or can be used as the basis for designing more ``practical'' schemes that approximate its functionality. For example, a large portion of the tomography literature is based on multicast, which has not been deployed in the Internet either,  thus one can say that it was not ``practical''. Nevertheless, multicast tomography showed how to exploit fundamental properties of topology-dependent correlation in multicast traffic for inference. Later on, unicast tomography used unicast probes to develop more ``practical'' schemes that mimicked and approximated the multicast tomography. Similarly, we believe that our work provides a fundamental building block for exploiting correlation in the content of network coded packets for inference. In our context, network coding can be thought of as ``reverse multicast'' that can be exploited for inference.

\subsection{Comparison to Traditional Tomography w/o Network Coding}
\label{sec-traditional}

Within the large literature on network tomography, the most closely related work is the Multiple Source Network Tomography in \cite{journal, probing, merging}, which formally defines M-by-N tomography problem. Our work on DAGs builds on \cite{journal, merging}: we follow their approach for decomposing the M-by-N into a number of 2-by-2 components, inferring the type of each 2-by-2 and then merging them together to reconstruct the M-by-N. Using simple network coding operations at intermediate nodes provides a graceful way to reveal coding points, which has been typically a challenge in traditional tomography.  Our work improves upon \cite{journal,merging} in that: (i) it can {\em exactly} identify the 2-by-2 type, as opposed to just distinguish between shared and non-shared types; and (ii) the merging algorithms can precisely locate the joining points with
respect to the branching points, as opposed to provide bounds.

Simulation results in Section \ref{sec-simul} on the same topology used in \cite{journal}, showed that our approach is more accurate, with less experiments. In essence, our approach is deterministic (one observation suffices to distinguish among types) as opposed to probabilistic (which needs to collect
a large number of probes for statistical significance). This benefit comes at the cost of having
intermediate nodes do some operations. However, these operations are so simple (just additions), that can be simply thought of as inverse multicast. This cost can be removed, if our approach is implemented as passive on top of random network coding, as outlined in Section \ref{sec-randomNC}.

\subsection{Comparison to Passive Tomography with Network Coding}
\label{sec-passive}

Recently, a passive approach for topology inference on top of random network coding has been proposed in \cite{jaggi}. Probes are sent once, and intermediate nodes pick coding coefficients $\beta$ uniformly at random out of a large field $\mathbf{F}_q$. The key idea is that, under assumptions of strong connectivity and large enough finite field, $\mathbf{F}_q$, the transfer matrix $M$, from the sender to the receiver, is distinct for different networks, w.h.p. Then, using the observations $Y$ at the receiver and the source messages $X$, exhaustive enumeration of all possible topologies is used to find an $M$ that matches $Y=MX$. An extended version of this work to erroneous networks is provided in \cite{jaggi-journal}, where different (ergodic or adversarial) failures lead to different transfer functions. Our approach is different in that it is active and uses several probes but simple coding operations over a small field.

\begin{example}
To better illustrate the differences, we consider a 2-by-2 topology, and we try to infer its type using the two approaches. The transfer matrices corresponding to the four 2-by-2 types of Fig.~\ref{fig-2by2} are provided in Fig. \ref{fig:equation}. The approach in \cite{jaggi, jaggi-journal} tries to distinguish among these four $M$'s in a single experiment. In contrast, we send probe packets in multiple rounds. In each experiment, $\beta$'s are either 0 or 1 (since we do additions only\footnote{In a joining point $J$, the $\beta$ from an incoming link to the outgoing link is 1 if the packet arrives at $J$ within $W$, and 0 otherwise. In a branching point $B$, all $\beta$'s are 1 unless there is loss, which makes $\beta$'s 0 in both $J$'s and $B$'s.}). We exclude some of the possible topologies in each experiment, until we are left with only one unique topology, in at most $countMax$ experiments.
\end{example}

\begin{figure*}[t!]
\footnotesize %\setcounter{MYtempeqncnt}{\value{equation}}
\vspace*{4pt} \hrulefill
\begin{displaymath} \label{eq_M1}
M_1=\left( \begin{array}{cc}
\beta_{S_1J,JB}\cdot\beta_{JB,BR_1} & \beta_{S_1J,JB}\cdot\beta_{JB,BR_2}\\
\beta_{S_2J,JB}\cdot\beta_{JB,BR_1} & \beta_{S_2J,JB}\cdot\beta_{JB,BR_2}\\
\end{array} \right)
\end{displaymath}
\begin{displaymath} \label{eq_M2}
M_2=\left( \begin{array}{ll}
\beta_{S_1J_1,J_1B_1}\cdot\beta_{J_1B_1,B_1R_1} & \beta_{S_1J_1,J_1B_1}\cdot\beta_{J_1B_1,B_1J_2}\cdot\beta_{B_1J_2,J_2R_2}\\
\beta_{S_2B_2,B_2J_1}\cdot\beta_{B_2J_1,J_1B_1}\cdot\beta_{J_1B_1,B_1R_1} & \beta_{S_2B_2,B_2J_2}\cdot\beta_{B_2J_2,J_2R_2}+\beta_{S_2B_2,B_2J_1}\cdot\beta_{B_2J_1,J_1B_1}\cdot\beta_{J_1B_1,B_1J_2}\cdot\beta_{B_1J_2,J_2R_2}\\
\end{array} \right)
\end{displaymath}
\begin{displaymath} \label{eq_M3}
M_3=\left( \begin{array}{ll}
\beta_{S_1J_2,J_2B_1}\cdot\beta_{J_2B_1,B_1J_1}\cdot\beta_{B_1J_1,J_1R_1} & \beta_{S_1J_2,J_2B_1}\cdot\beta_{J_2B_1,B_1R_2}\\
\beta_{S_2B_2,B_2J_1}\cdot\beta_{B_2J_1,J_1R_1}+\beta_{S_2B_2,B_2J_2}\cdot\beta_{B_2J_2,J_2B_1}\cdot\beta_{J_2B_1,B_1J_1}\cdot\beta_{B_1J_1,J_1R_1} & \beta_{S_2B_2,B_2J_2}\cdot\beta_{B_2J_2,J_2B_1}\cdot\beta_{J_2B_1,B_1R_2}\\
\end{array} \right)
\end{displaymath}
\begin{displaymath} \label{eq_M4}
M_4=\left( \begin{array}{cc}
\beta_{S_1B_1,B_1J_1}\cdot\beta_{B_1J_1,J_1R_1} & \beta_{S_1B_1,B_1J_2}\cdot\beta_{B_1J_2,J_2R_2}\\
\beta_{S_2B_2,B_2J_1}\cdot\beta_{B_2J_1,J_1R_1} & \beta_{S_2B_2,B_2J_2}\cdot\beta_{B_2J_2,J_2R_2}\\
\end{array} \right)
\end{displaymath}
%\setcounter{equation}{\value{MYtempeqncnt}}
%\vspace*{4pt}
\hrulefill
\caption{Comparison of our approach to \cite{jaggi,jaggi-journal} in the example of a 2-by-2 topology. $M_1$, $M_2$, $M_3$, and $M_4$ are the transfer matrices resulting from the four different types (types 1, 2, 3, and 4, respectively, in Fig. \ref{fig-2by2}) of a 2-by-2, if intermediate nodes use coding coefficients $\beta$. The approach in \cite{jaggi, jaggi-journal} tries to distinguish among these four $M$'s in a single experiment. In contrast, we use $\beta$'s either 0 or 1 and multiple experiments to choose an $M$.}\label{fig:equation}
%\vspace{-0.5em}
\end{figure*}

Our $countMax$ experiments can be thought of as collecting observations $Y_1=M_1X, Y_2=M_2X,\cdots, Y_{countMax}=M_{countMax}X$, where $M_1,M_2,\cdots,M_{countMax}$ are different representations of the unique $M$. Note that, although $M$ is unique in terms of $\beta$'s for each topology, it can be shown to be non-unique when these $\beta$'s are replaced by $0/1$ values. For example in Fig. \ref{fig:equation}, the transfer matrices of types 1 and 4 would look similar if all $\beta$'s are equal to 1, but only type 4 can potentially result in $M=[1,1;0,1]$, while type 1 cannot. We send probes in multiple experiments to create those representations of $M$ that help us uniquely identify the underlying topology.

In terms of the field size, \cite{jaggi, jaggi-journal} require a larger field than us, to obtain distinct transfer matrices for different topologies.\footnote{\cite{jaggi} shows that if local coding variables are i.i.d uniform random variables over $\mathbf{F}_q$, then the probability that all
different unicast networks with at most $|V|$ nodes and $|E|$ edges have distinct transfer matrices is
$\ge 1-|V|^{4|E|}(1-(1-\frac{1}{q})^{|V|})$. This shows that: (i) the success probability $\rightarrow1$ iff $q\rightarrow\infty$; (ii) $q$ needs to increase rapidly as the network size grows. For example, our approach described in Section~\ref{sec-2by2} requires a small field $\mathbf{F}_{3}$ to distinguish among different 2-by-2's. While one can calculate that if $\beta$'s are chosen uniformly at random out of $\mathbf{F}_{3}$, then $Pr(M_4=M_1)\cong0.04$.}
In terms of bandwidth, we use small packets in each experiment, since we perform operations over a small field, and a few experiments are required. On the other hand in \cite{jaggi}, the coefficients, sent anyway along with packets through the network, are used to reveal the topology from the transfer matrix at the receiver, and can be thought of as the equivalent of probes. The distinction between active and passive approaches becomes even less pronounced, if we consider that \cite{jaggi, jaggi-journal} require the receiver to have a-priori knowledge of the size of the network, and of the code-book used at each node (referred to as {\em common randomness}), which depends on the node id \cite{jaggi-journal}. We do not require such knowledge and we infer the
topology with smaller complexity. 
Further comparison of our  approach to \cite{jaggi-journal} for larger M-by-N topologies can be found in \cite{active-topology-arxiv}.

\subsection{Extension to Passive Tomography with RNC} \label{sec-randomNC}

We now discuss how our approach can potentially be extended to be implemented
as passive, when random network coding (RNC) is used in the middle. Intuitively, the same topology inference algorithms should apply if we ensure that RNC coefficients satisfy necessary conditions for inference.

Assume that in each experiment, the intermediate nodes perform random linear network
coding operations instead of the simple additions assumed so far. In this case, Alg.~\ref{alg-generalLossy} will still work if we assign coding coefficients to the joining points in a partial order, so as to ensure that the minimum coding coefficient of a joining point is always greater than or equal to
the maximum coding coefficient of its ancestor joining point. Under this condition, we can prove that the same rationale as in Section~\ref{sec-2by2Lossy} still holds, \ie type 1 results in similar observations; type 2 results in $c_{12}-c_{22}\leq0$; type 3 results in $c_{12}-c_{22}\geq0$; and type 4 results in $c_{12}-c_{22}\leq0$ or $c_{12}-c_{22}\geq0$. The proof can be found in the technical report \cite{active-topology-arxiv}. Code design to jointly meet both random network coding goals (large enough field for independent linear equations) and tomographic goals (the aforementioned condition) is part of future work. If such a code design is possible, Algorithm \ref{alg-generalLossy} can be directly applied to the case of random network coding. An example can be found in \cite{active-topology-arxiv}.

\subsection{Comparison to {\tt traceroute}-like Approaches}\label{sec-traceroute}

In practice, the dominant approach to Internet mapping is based on {\tt traceroute} \cite{govindan,cheswick,rocketfuel,yao,clauset,dallasta,skitter,doubletree,paris,DIMES}.
It uses {\tt traceroute}'s sent between selected nodes and collects
the ids of the nodes along the paths traversed. It faces the challenges of (i) resolving anonymous routers and router aliases and (ii) causing congestion close to the monitoring points {\cite{doubletree}.

Similarly to {\tt traceroute}, we also use active end-to-end probes and we require some minimal co-operation from internal nodes (simple additions in our case vs. traceroute-specific responses). However, unlike {\tt traceroute}, we do not ask intermediate nodes to reveal their node id, which has the advantage of preserving the anonymity of intermediate nodes. A design difference was also noted in Section \ref{sec-otherStructures}: we infer 2-by-2 components, instead of 1-by-1's (paths) for {\tt traceroute}.

In terms of measurement bandwidth, our approach uses exactly one probe per link per experiment, which is the minimum possible. This is thanks to network coding that combines multiple incoming packets into one, and thanks to multicast that replicates a single incoming packet into many outgoings, thus eliminating overlap. On the other hand, standard {\tt traceroute} implementations, \eg {\em skitter} \cite{skitter}, send three probe packets for each hop count in every single path.\footnote{As an example, the average number of packets/link required by both standard {\tt traceroute} \cite{cheswick} and our approach to identify the 2-by-2's is as follows: in type 2 or 3, it is 12 for {\tt traceroute} and 1 for our approach. In type 1, it is 14.4, and in type 4, it is 9, for {\tt traceroute}, while both are $countMax$ for our approach. The benefit of our approach in terms of the number of required packets depends on the topology. Also, we can trade-off accuracy for the load by adjusting $countMax$.
On the other hand, {\em Doubletree} \cite{doubletree} can reduce the {\tt traceroute} overhead to some extent.} For example, our approach reduces the {\em number of measured paths} in a 2-by-N topology by a factor of two, compared to {\tt traceroute}; {\em i.e.,} we require $O(N)$ instead of $O(2N)$ measurements, since each coded packet observes two paths.

%\eg to 10.8 packets/link in type 1 and 6.75 in type 4.

%\footnote{\textcolor{blue}{However, in M-by-N topologies, {\tt traceroute} requires $O(MN)$ measurements, while tomographic techniques such as our work and  \cite{merging} require $O(\binom{M}{2}N)$ measurements. We are currently working on improving the merging algorithm.}}

%A quantitative comparison of our approach vs. {\tt traceroute}  is provided in the technical report  \cite{active-topology-arxiv}.

In terms of the amount of information per probe, we can think of our approach as {\tt traceroute} probes that only report the joining points in the topology; while {\tt traceroute} reports all the node ids, some of them being duplicate. Therefore, the amount of information per probe is significantly decreased in our approach. In essence, we only report the minimum required information, which we prove to be sufficient for reconstructing the topology.
%In addition, unlike {\tt traceroute}, our approach does not require the intermediate nodes to reply to the source with ICMP messages.

\subsection{Comparison to  Packet Marking Approaches}
\label{sec-traceback}

As we mentioned in the previous section, {\tt traceroute} sends many rounds of probes, one for each hop towards the destination \cite{govindan,cheswick,rocketfuel,yao,clauset,dallasta,skitter,doubletree,paris,DIMES}. Therefore, our approach is not directly comparable to {\tt traceroute}. On the other hand, as we described in Section~\ref{sec-deployment}, our scheme can also be implemented as a packet marking scheme on top of regular unicast flows. Indeed, the most comparable approaches to our work are the packet marking schemes, \eg {\em traceback} schemes \cite{savage}, which aim at identifying the source(s) of a sequence of packets and the nodes these packets have traversed. This is useful for tracing the source(s) of high volume traffic, \eg in Distributed Denial-of-Service attacks \cite{savage}. In packet marking schemes, intermediate nodes (probabilistically) mark packets with information about their identity, and the receiver uses the information from several packets to reconstruct the path(s) they have traversed. Packet marking schemes have been mostly used in a single-path or a reverse tree topology, rooted on the victim of  the attack. Therefore, they identify an M-by-1 topology, while we identify a more general M-by-N topology.

On the other hand, traceback needs to reconstruct not only the node ids, but also the order in which they are traversed by the packets, \ie the attack paths and graph. In single-path scenarios, this can be done by adding a distance value to the marking field. However, in multi-path scenarios, this has been a challenge since there exist multiple nodes at the same distance from the receiver. Many approaches have been proposed to overcome this limitation, and most of them suffer from practical issues in terms of the amount of space required in the packet header for marking. In contrast, multi-path scenarios are the strength of our approach, which identifies the joining points using network coding.

\section{Conclusion}
\label{sec-conclusion}

In this paper, we designed active probing schemes that exploit simple operations at intermediate nodes to accurately infer the topology, based on end-to-end observations. We designed algorithms for trees and general topologies, and we simulated them in representative examples. 
Our schemes build on the work by Rabbat {\em et al.} \cite{journal,probing,merging} and extend it when joining points perform network coding operations. Furthermore, we make connections with several alternative approaches, including passive inference, {\tt traceroute}, and packet marking.  Our main contribution is that we show how to exploit the fundamental connection between network coding and topology, and thus adding one new building block in the, already large, space of available options for topology inference. We expect the techniques developed in this paper to be most useful in networks that can perform simple network coding operations, including but not limited to wireless multi-hop networks.

\section*{Acknowledgment}
We would like to thank Prof. Animashree Anandkumar at UC Irvine for useful discussion on the performance analysis of Algorithm~\ref{alg-generalLossy} and on the complexity of merging algorithms.

\appendix
\section{Probability of Error of Alg.~\ref{alg-binaryLossy} (Topology Inference in Trees with Loss)}
\label{appendix-Alg2}

In this Appendix, we consider the topology inference algorithm for tree networks in the presence of packet loss (\ie Alg.~\ref{alg-binaryLossy}), and we derive the exact probability of error of its first iteration in inferring the binary tree topology of Fig.~\ref{fig-binaryExample}. Assume that all links have the same loss rate $p$. Let $X_i$, $i=2,3,4,5,6$, be the number of experiments required for receiver $i$ to observe one correct packet in the first iteration, as described in Example~\ref{ex:binarytree}; \eg$X_2$ is the number of experiments required for receiver 2 to observe $x_1$, $X_3$ is the number of experiments required for receiver $3$ to observe $x_1+x_2$, etc. $X_i$ is a geometric random variable with probability of success $\rho_{i}$, as follows: 

%Similarly, let $X_5, X_6$ denote the number of experiments required for receivers $5,6$, respectively, to receive $x_2$, and let $X_3, X_4$ denote the number of experiments required for receivers $3,4$, respectively, to receive $x_1+x_2$. 

\begin{equation}
\rho_{2}={(1-p)}^2 \, , \, \rho_{5}=\rho_{6}={(1-p)}^3 \, , \, \rho_{3}=\rho_{4}={(1-p)}^6
\end{equation}

If we denote the total number of required experiments for the first iteration of Alg.~\ref{alg-binaryLossy} by $X$, then we have that: $X=max_{i\in\{2,3,4,5,6\}}{X_i}$. Therefore, we can compute the Chebyshev bound on $X$. On the other hand, in $M$ experiments, we can also compute the probability of error as follows:

\begin{equation}
\label{eq-errorExact-Alg2}
Pr(error) = 1- Pr(X_2\leq M \& X_3\leq M \& X_4\leq M \& X_5\leq M \& X_6\leq M)
\end{equation}

For general $M$, Eq.(\ref{eq-errorExact-Alg2}) becomes complicated to compute. Here we only compute it for $M=1$:

\begin{equation}
\label{eq-errorExact2-Alg2}
\begin{split}
Pr(error) &= 1-[Pr(X_2=1)Pr(X_5=1|X_2=1) Pr(X_6=1|X_2=1,X_5=1) \\
&\quad Pr(X_3=1|X_2=1,X_5=1,X_6=1) Pr(X_4=1|X_2=1,X_5=1,X_6=1,X_3=1)]\\
&\quad = 1- {(1-p)}^2 {(1-p)}^3 {(1-p)} {(1-p)}^4 {(1-p)}
\end{split}
\end{equation}

For example, for $M=1$ and $p=10\%$, Eq.(\ref{eq-errorExact2-Alg2}) results in $Pr(error)=0.69$, which agrees with the simulation results in Fig.~\ref{fig-binaryTreeSimul}(a). We can also provide an upper bound on $Pr(error)$ using the union bound, as follows:

\begin{equation}
\label{eq-errorUnionBound-Alg2}
Pr(error) \leq \sum_{i\in\{2,3,4,5,6\}} Pr(X_i>M) = (1-{(1-p)}^2)^M+2(1-{(1-p)}^3)^M+2(1-{(1-p)}^6)^M
\end{equation}

For example, for $M=5$ and $p=10\%$, Eq.(\ref{eq-errorUnionBound-Alg2}) results in $Pr(error)\leq 0.05$, which is an upper bound on the probability of error that we observe in Fig.~\ref{fig-binaryTreeSimul}(a). Similar analysis can be used for the second iteration and for any other tree topology.

\section{Prob. of Error of Alg.~\ref{alg-generalLossy} (Inference of 2-by-2's in DAGs with loss)}
\label{appendix-Alg6}

In this Appendix, we consider the topology inference algorithm for DAGs in the presence of packet loss (\ie Alg.~\ref{alg-generalLossy}), and we derive the probability of error of Alg.~\ref{alg-generalLossy}, which we stated in Lemma~\ref{thm-Thm3}. Assume that all links have the same loss rate $p$. As we mentioned in Section~\ref{sec-generalSimulSetup}, an error in Alg.~\ref{alg-generalLossy} can result either from declaring type 4 as type 2 or 3, or from declaring type 2 or 3 or 4 as type 1. In $countMax$ experiments, the probability of error is maximized when the underlying topology is of type 4. In fact, a type 4 topology can give an upper bound on $countMax$, as it requires the maximum number of experiments to be identified correctly: it requires both a $c_{12}-c_{22}>0$ observation and a $c_{12}-c_{22}<0$ observation. 

Let $Y_1$ be the number of experiments required to obtain a $c_{12}-c_{22}>0$ observation and let $Y_2$ be the number of experiments required to obtain a $c_{12}-c_{22}<0$ observation. Both $Y_1$ and $Y_2$ are geometric random variables with success probabilities $\rho'_{1}$ and $\rho'_{2}$, respectively. In a type 4 topology, from Table~\ref{table-lossyObs}, we have that:

\begin{equation}
\rho'_{1} = Pr(x_1+x_2,x_1) + Pr(x_2,x_1)
\end{equation}

One can compute each of these probabilities by considering the loss events and the late arrival events (at joining points) that can cause such observations in type 4. If we denote the probability that packet $x_2$ arrives at joining point $J_i$ within $W$ by $Pr(u+D_i\leq W)=\gamma_i$, $i=1,2$, and we also use the notation $Pr(u+D_i>W)=\overline \gamma_i$, then we have that:

\begin{equation}
Pr(x_1+x_2,x_1) = p{(1-p)}^7 \gamma_1+{(1-p)}^8 \overline \gamma_2 \gamma_1 = {(1-p)}^7 \gamma_1 (p+(1-p) \overline \gamma_2)
\end{equation}

Also:

\begin{equation}
Pr(x_2,x_1) =p^2 {(1-p)}^6 + {(1-p)}^7 \overline \gamma_2 p = p{(1-p)}^6 (p+(1-p) {\overline \gamma}_2) 
\end{equation}

Therefore, we have that:

\begin{equation}
\rho'_{1} = {(1-p)}^6 (p+(1-p) \gamma_1) (p+(1-p) \overline \gamma_2) 
\end{equation}

Similarly, one can compute that:

\begin{equation}
\rho'_{2} = {(1-p)}^6 (p+(1-p) \gamma_2) (p+(1-p) {\overline \gamma}_1)
\end{equation}

If we denote the total number of required experiments for Alg.~\ref{alg-generalLossy} by $Y$, then we can estimate $Y$ by $Y\leq Y_1+Y_2$. If the link delay distribution and the link loss rates are given, we can use the above relationships to find the Chebyshev bound on $Y$. On the other hand, in $countMax$ experiments, we can also compute the probability of error as follows:

\begin{equation}
Pr(error) = 1- Pr(Y_1\leq countMax \& Y_2\leq countMax)
\end{equation}

We can assume that the events $Y_1\leq countMax$ and $Y_2\leq countMax$ are independent. Thus:

\begin{equation}
Pr(error) = 1- \sum_{i=0}^{countMax-1} (1-\rho'_{1})^i \rho'_{1} \sum_{j=0}^{countMax-1} (1-\rho'_{2})^j \rho'_{2}
\end{equation}

We can also provide an upper bound on $Pr(error)$ using the union bound, as follows:

\begin{equation}
\label{eq-errorUnionBound-Alg6}
Pr(error)\leq Pr(Y_1>countMax) + Pr(Y_2>countMax) = (1-\rho'_{1})^{countMax} + (1-\rho'_{2})^{countMax}
\end{equation} 

Therefore, the proof of Lemma~\ref{thm-Thm3} is complete.
\hfill{$\square$}

%Note that we can approximately assume that $\gamma_1 \cong \gamma_2$, and therefore $\rho'_{1} \cong \rho'_{2}=\rho'$. In that case, we can simplify Eq.(\ref{eq-errorUnionBound-Alg6}) as $Pr(error)\leq 2(1-\rho')^{countMax}$. 

%\bibliographystyle{model1-num-names}
%\bibliography{references}

\input{bios}

\end{document}

%% file: bios.tex
%\vspace{5pt}

%\begin{small}

\parpic{\includegraphics[width=1in,clip,keepaspectratio]{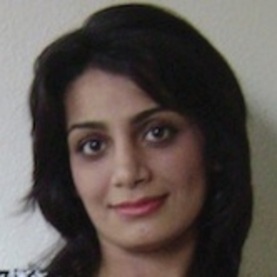}}
\noindent {\bf Pegah Sattari} (SM'08) received the B.S. degree in Electrical Engineering from Sharif University of Technology, Tehran, Iran, in 2006, and the M.S. degree in Electrical and Computer Engineering from the University of California, Irvine, in 2007. She is currently a Ph.D. candidate in the EECS Department at the University of California, Irvine. Her research interests include network measurements, network coding, and its applications to inference problems. %\vspace{1.0em}

\parpic{\includegraphics[width=1in,clip,keepaspectratio]{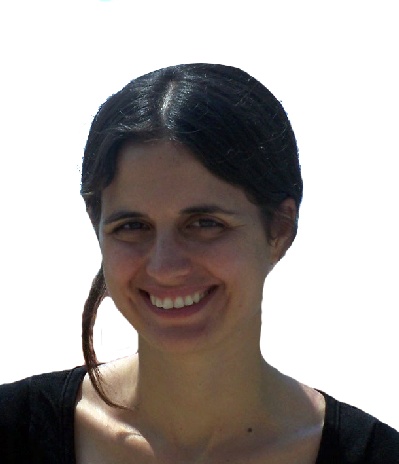}}
\noindent {\bf Christina Fragouli} is an assistant professor in the School of Computer and Communication Sciences, EPFL, Switzerland. She received the B.S. degree in Electrical Engineering from the National Technical University of Athens, Greece, in 1996, and the M.Sc. and Ph.D. degrees in Electrical Engineering from the University of California, Los Angeles, in 1998 and 2000, respectively. She has worked at the Information Sciences Center, AT\&T Labs, and the National University of Athens. She has also visited Bell Labs and DIMACS, Rutgers University. Her research interests include network coding, network information flow theory and algorithms, and connections between communications and computer science. She received the ERC Starting Grant from the European Research Council in 2009.
%the Outstanding Ph.D. Student Award 2000-2001, UCLA, the Zonta award 2008 in Switzerland, and 

\parpic{\includegraphics[width=1in,clip,keepaspectratio]{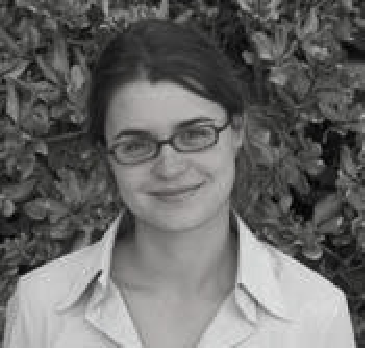}}
\noindent {\bf Athina Markopoulou} (SM'98, M'02) is an assistant professor in the EECS Department at the University of California, Irvine. She received the Diploma degree in Electrical and Computer Engineering from the National Technical University of Athens, Greece, in 1996, and the M.S. and Ph.D. degrees in Electrical Engineering from Stanford University in 1998 and 2003, respectively. She has been a postdoctoral fellow at Sprint Labs and at Stanford University, and a member of the technical staff at Arastra Inc.. Her research interests include network coding, network measurements and security, media streaming and online social networks. She received the NSF CAREER award in 2008.

%\end{small}